\documentclass[a4paper,11pt]{article}
\usepackage{amsfonts,amssymb,amsmath,amsthm,dsfont}

\usepackage{subfigure}
\usepackage{fullpage, color, upgreek}
\usepackage{graphicx}
\makeatletter
\let\NAT@parse\undefined
\makeatother
\usepackage[sort&compress, numbers]{natbib}
\graphicspath{{./},{images/}}

\makeatletter
\newcommand{\pushright}[1]{\ifmeasuring@#1\else\omit\hfill$\displaystyle#1$\fi\ignorespaces}
\newcommand{\pushleft}[1]{\ifmeasuring@#1\else\omit$\displaystyle#1$\hfill\fi\ignorespaces}
\makeatother

\usepackage{hyperref}
\newcommand{\ts}{\textstyle}
\newcommand{\Rbb}{\mathbb{R}}

\newcommand{\Nbb}{\mathbb{N}}
\newcommand{\Zbb}{\mathbb{Z}}
\newtheorem{theorem}{Theorem}

\newtheorem{corollary}{Corollary}
\newtheorem{lemma}{Lemma}
\newtheorem{remark}{Remark}

\newcommand{\inv}[1]{\frac{1}{#1}}
\newcommand{\supp}{{\rm supp}\,}
\newcommand{\tinv}[1]{{\textstyle\frac{1}{#1}}}

\newcommand{\ud}{\mathrm{d}}

\renewcommand{\leq}{\leqslant}
\renewcommand{\geq}{\geqslant}

\DeclareMathOperator{\Id}{\mathds{1}}

\DeclareMathOperator*{\argmin}{argmin}

\newcommand{\bb}{\mathbb}

\newcommand{\bs}{\boldsymbol}
\newcommand{\cl}{\mathcal}
\newcommand{\ie}{\emph{i.e.}, }
\newcommand{\eg}{\emph{e.g.}, }

\newcommand{\bsxalt}{{\bs x}^*}

\title{\huge Error Decay of (almost) Consistent Signal Estimations\\
from Quantized Gaussian Random Projections}
\author{Laurent Jacques\thanks{LJ is with the ISPGroup, ICTEAM institute, ELEN Department, Universit\'e catholique de Louvain
(UCL), Belgium. Email: \url{laurent.jacques@uclouvain.be}. LJ is funded by Belgian National Science Foundation
(F.R.S.-FNRS).}}

\begin{document}

\maketitle

\begin{abstract}
This paper provides new error bounds on \emph{consistent}
reconstruction methods for signals observed from quantized random
projections. Those signal estimation techniques guarantee a perfect matching
between the available quantized data and a new observation of
the estimated signal under the same sensing model.
Focusing on dithered uniform scalar
quantization of resolution~$\delta>0$, we prove first that, given a
Gaussian random frame of~$\bb R^N$ with~$M$ vectors, the worst-case~$\ell_2$-error of consistent signal reconstruction
decays with high
probability as~$O(\tfrac{N}{M}\log\tfrac{M}{\sqrt N})$ uniformly for all
signals of the unit ball~$\mathbb B^N \subset \mathbb R^N$. Up to a
log factor, this matches a known lower bound in
$\Omega(N/M)$ and former empirical validations
  in~$O(N/M)$. Equivalently, if $M$ exceeds a minimal number of frame coefficients
growing like~$O(\tfrac{N}{\epsilon_0}\log \tfrac{\sqrt
  N}{\epsilon_0})$, any vectors in~$\mathbb B^N$ with~$M$ identical quantized
projections are at most~$\epsilon_0$ apart with high probability.
Second, in the context of Quantized Compressed Sensing with~$M$
Gaussian random measurements and under the same scalar quantization scheme, consistent
reconstructions of~$K$-sparse signals of~$\bb R^N$ have a worst-case
error that decreases with high probability as~$O(\tfrac{K}{M}\log\tfrac{MN}{\sqrt K^3})$
uniformly for all such signals.
Finally, we show that the proximity of vectors whose quantized random
  projections are only approximately consistent can still be bounded with
  high probability. A certain level of corruption is thus allowed in the quantization
  process, up to the appearance of a systematic bias in the reconstruction error of (almost)
  consistent signal estimates.
\end{abstract}

\section{Introduction}
\label{sec:introduction}
Since the advent of the digital signal processing era and of analog-to-digital converters, an intense field of
research has been concerned by the non-linear sensing model
\begin{equation}
  \label{eq:qcs-model}
  \bs q = \cl Q[\bs A \bs x] \in \cl J,  
\end{equation}
where~$\bs A \in \bb R^{M\times N}$ is a matrix representing a
linear transformation of a signal~$\bs x$ taken in some bounded subset
$\cl K$ of~$\bb R^N$, and~$\cl Q$
stands for a quantization of~$\bs A \bs x$ that maps~$\bs A\cl K :=
\{\bs A\bs u: \bs u \in \cl K \} \subset
\bb R^M$ to a
finite set of vectors~$\cl J \subset\Rbb^M$, \eg encoded over a given number of bits
\cite{Gray98,BoufChapter}. 

The bounded space~$\cl K$ containing generally 
an infinite number of signals, the model~\eqref{eq:qcs-model} is of course lossy and~$\bs
x$ cannot be recovered exactly from~$\bs q$. Quantifying this loss of
information as a function of both the signal reconstruction method and
of the key elements~$\bs A$,~$N$,~$M$,~$\cl K$ and~$\cl Q$ has
been therefore the topic of many studies at the frontier of
information theory, high-dimensional geometry, signal processing and statistics.

The general model~\eqref{eq:qcs-model} is
for instance the one adopted in Quantized Compressed Sensing (QCS)
\cite{candes2006near,BoufChapter,Jacques2010}, where the signal~$\bs x$ is assumed
sparse or compressible in an orthonormal basis~$\bs \Psi$ of~$\Rbb^N$, and the
sensing matrix is generated randomly, \eg from Gaussian random
ensembles~\cite{candes2006ssr}. When~$M\geq N$,
Eq.~\eqref{eq:qcs-model} is also a model for \emph{frame coefficient quantization} (FCQ) of signals in~$\bb R^N$, \ie when
the coefficients~$\bs A \bs x = (\bs a_1^T\bs x, \cdots, \bs a_M^T \bs
x)^T$ of~$\bs x$ in an overcomplete frame of
$\bb R^N$ are quantized in~$\bs q = \cl Q(\bs A \bs x) \in \cl J$,~$\bs A=(\bs a_1, \cdots, \bs a_M)^T$ representing the matrix whose
row set~$\{\bs a_j \in \bb R^N: 1\leq j\leq M\}$
collects the frame vectors~\cite{goyal_1998_lowerbound_qc,goyal2001quantized,bib:Boufounos06thesis}. 

In this work, we restrict the analysis of~\eqref{eq:qcs-model} to a scalar, regular and uniform
quantizer. The quantization~$\cl Q$ is then a scalar operation
applied componentwise on vectors; its 1-D quantization
cells~$\cl Q^{-1}[c] = \{\lambda: \cl Q[\lambda] = c\}\subset \Rbb$ are
convex and have all
the same size (or \emph{resolution})~$\delta > 0$. Other
quantization procedures have been studied for~\eqref{eq:qcs-model} and we refer the reader for instance to~\cite{BoufChapter}
for a review of scalar and~$\Sigma\Delta$-quantization
\cite{gunturk2013sobolev} in the QCS literature, to~\cite{B_TIT_12}
for a theoretical analysis of non-regular 
scalar quantizers, to~\cite{pai_nonadapt_MIT06,kamilov_2012}
for the use of non-regular \emph{binned} quantization, or to
\cite{vivekQuantFrame} for an example of vector quantization by frame
permutation. As realized in
\cite{gunturk2013sobolev,goyal_2008_cs_lossycomp,BoufChapter}, we also assume that the ``variability'' of the components
of~$\bs A \bs x$, also measured by their variance, does not change with~$M$. This is critical for defining
a quantizer~$\cl Q$ of constant resolution when~$M$ increases. 

Many studies have addressed the model~\eqref{eq:qcs-model}
by observing that the distortion induced by quantization 
compared to a linear model~$\bs A\bs x$ is the one of an additive measurement noise~$\bs n=\cl
Q[\bs A \bs x] - \bs A\bs x$ with~$n_i
\in [-\delta/2,\delta/2]$, \ie
\begin{equation}
  \label{eq:noisy-CS}
  \bs q = \bs A \bs x + \bs n.  
\end{equation}
When the resolution is small compared to the standard deviation of
  each component of~$\bs A \bs x$, \ie under the \emph{high resolution assumption}
\cite{Gray98,BoufChapter,Jacques2010}, or if a random \emph{dithering} is
added prior to quantization~\cite{Gray98}, each component of the noise
can be assumed as
uniformly distributed within~$[-\delta/2,\delta/2]$. This allows one
to bound the
power of this noise, \ie $\bb E(\|\bs n\|^2) = M \delta^2/12$ and~$\|\bs n\|^2
\leq \tinv{12}\,\delta^2 (M + \zeta \sqrt M)$ with high probability
for~$\zeta = O(1)$ (see, \eg~\cite{Jacques2010}).  

In the case of QCS, when a general noise~$\bs n$ of bounded
power~$\|\bs n\| \leq \varepsilon$ corrupts the compressive
observation of a sparse signal~$\bs x$ as in~\eqref{eq:noisy-CS}, a
worst-case reconstruction error that follows  
$$
\|\bs x - \bsxalt\| =
O(\varepsilon/\sqrt M),
$$
can be reached by various reconstruction methods (\eg Basis Pursuit
DeNoise~\cite{Chen98atomic,candes2006ssr} or Iterative Hard
Thresholding~\cite{BluDav::2008::Iterative-hard})
as soon as a suitably rescaled sensing matrix
$\inv{\sqrt M}\bs A$ respects the restricted isometry property
(RIP)~\cite{candes2006ssr}.

Thus, when the compressive observations of a~$K$-sparse signal undergo
uniform scalar quantization, it is then expected
that, with high probability,~$\|\bs x - \bsxalt\| = O(\delta)$ by
setting~$\varepsilon^2 = \tinv{12}\,\delta^2 (M + \zeta \sqrt
M)$. The constancy of this error with respect to~$M$ is also known as the
classical error limit of the \emph{pulse code modulation scheme} (PCM)  in CS~\cite{gunturk2013sobolev}. 

However, most of the reconstruction techniques enforce a
$\ell_2$-norm fidelity with~$\bs q$, \eg by imposing~$\|\bs A \bsxalt - \bs q\| \leq
\varepsilon$, and the reconstructed signal is not guaranteed to be
consistent with the observations, \ie $\cl Q[\bs A \bsxalt] \neq
\bs q$. The knowledge of the sensing model is thus not fully exploited
for reconstructing~$\bs x$ from~$\bs q$.

In the context of signal representations using frames, it is also
known that, for unit-norm frame vectors, if the frame coefficients of a signal are corrupted by an
additive noise of variance~$\sigma^2$, the linear signal estimate
synthesized from the dual frame on these coefficients has a root mean square error (RMSE) lower bounded by~\citep{goyal2001quantized,powell_consistent}
$$
(\bb E\|\bs x - \bsxalt\|^2)^{\inv{2}} \geq N \sigma /\sqrt M,
$$  
where the expectation is taken with respect to the noise. This
shows that for FCQ the
reconstruction error decay of such a linear reconstruction is limited
to~$O(N\delta/\sqrt M)$ since~$\sigma = O(\delta)$. Nevertheless, the produced
solution is also inconsistent with the observations, \ie $\bs A
\bsxalt \notin \cl Q^{-1}[\bs q]$, as the signal synthesis reached by the
dual frame~$\bs A^\dagger$ amounts to solving the
least-squares problem~$\bsxalt = \argmin_{\bs u\in\bb R^N} \|\bs q - \bs A\bs u\|^2 = (\bs A^T \bs A)^{-1} \bs A^T
 \bs q = \bs A^\dagger \bs q$, which promotes an~$\ell_2$-norm
 fidelity with respect to~$\bs q$. 

This work studies a better approach for improving the reconstruction
error decay in both QCS and FCQ. The proposed analysis explicitly
enforces quantization consistency while reconstructing the signal, \ie
  finding an estimate~$\bsxalt \in \cl K$ such that~$\cl Q[\bs A\bsxalt]
  = \cl Q[\bs A\bs x]$. This procedure was initially introduced
in~\cite{Thao_Consistency_1994} for oversampled analog-to-digital
conversion of bandlimited signals, or in~\cite{goyal_1998_lowerbound_qc}
in the more general context of quantized
overcomplete signal expansion. In more detail,~\cite{goyal_1998_lowerbound_qc}~showed that, given a
random model on the generation of the sensed signal, the RMSE of any
reconstruction method is lower bounded by
$\Omega(N/M)$. Interestingly, the same
lower bound can also be obtained on the worst-case
reconstruction error without requiring any random model on the
source~\cite{BoufChapter}. While conjectured for general
frames with redundancy factor~$M/N$, the combination of a tight frame formed by an oversampled
Discrete Fourier Transform (DFT) with a
consistent signal reconstruction reaches this lower bound, \ie in this
case the RMSE is upper bounded by~$O(N/M)$~\cite{goyal_1998_lowerbound_qc}.
Numerically, the reconstruction errors of recovery methods based on alternate projections
onto convex sets\footnote{This method finds a
  vector of~$\bb R^M$ by alternate projections between the image of~$\bs A$ and the
  consistency cell~$\cl Q^{-1}[\bs q]$, these two sets being
  convex. The signal is reconstructed from the POCS solution using the dual frame.} (POCS)
\cite{Thao_Consistency_1994,goyal_1998_lowerbound_qc} or on message
passing algorithms~\cite{kamilov_2012}, have also been
observed to approach~$O(N/M)$ for both deterministic and random~$\bs A$.
Moreover, the Rangan-Goyal recursive algorithm
\cite{rangan_rec_algo_2001}, which enforces local consistency of the
current estimate at every iteration, provides reconstruction error
decaying as~$\bb E(\|\bs x-\bsxalt\|^2) = O(1/M^2)$ for random
frames, where the expectation is made on the (uniform) quantization
noise~\cite{powell_RGalgo}. 
More recently, Powell and Whitehouse~\cite{powell_consistent} have
analyzed geometrically the worst-case error of any consistent
reconstruction method. As detailed in Sec.~\ref{sec:prior-works}, they
showed for instance that, for frames constructed by taking~$M$ vectors picked uniformly at random
over the unit sphere~$\bb S^{N-1} \subset \bb R^N$, the expectation of
this error with respect to the random frame construction decays as~$O(\delta N^{3/2}/M)$. 

We can also mention that consistent reconstruction methods have been applied to QCS in the high-resolution regime (\ie for~$\delta \ll 1$ when~$\cl K \subset \bb B^N$)
\cite{Dai2009,Jacques2013,BoufChapter}; for uniform (or bounded) noise
for FCQ~\cite{powell_consistent}; in the extreme 1-bit QCS setting,
where quantization reduces to the application of a sign operator
\cite{BB_CISS08,jacques2013robust,plan2013one,plan2011dimension}; and
even for non-regular quantization scheme in QCS~\cite{B_TIT_12,BR_DCC13}. 

Compared to these former works, we highlight three important features of
this paper. These are only sketched in this Introduction and we refer
the reader to Sec.~\ref{sec:main-results} for their precise statements. 

First, we analyse the FCQ and the QCS contexts when~$\bs A$
is generated as a Gaussian random matrix and when the scalar
quantization incorporates a uniform \emph{dithering}
\cite{Gray98,B_TIT_12}. As will be clear later, this dithering, which
is often used to improve the statistical properties of the quantizer
and is assumed to be known in signal reconstruction, allows
us to leverage a geometric connection between quantized Gaussian random
projections of vectors and Buffon's needle problem~\cite{jacques2013quantized}.  

Second, we provide upper bounds for the worst-case reconstruction
error of consistent signal estimations, \ie valid for the reconstruction of any vector in the signal set~$\cl
K$. Those bounds hold with high (and controlled) probability over the generation of
both~$\bs A$ and the dithering as soon as~$M$ is large compared to the
\emph{complexity} of the signal space~$\cl K$. 

For instance, in the case of
FCQ with Gaussian random frames and~$\cl K = \bb B^N$, we show that
if~$M$ is bigger than a minimal value growing like
$O(\frac{\delta}{\epsilon_0} N \log \sqrt N/\epsilon_0)$, then, with
high probability, the distance between any pair of vectors in~$\cl K$ having consistent
observations through the mapping~\eqref{eq:qcs-model} cannot be larger
than~$\epsilon_0$ (see Theorem~\ref{thm:consist-impose-closeness}). Inverting the relationship between~$\epsilon_0$ and
the minimal~$M$, this establishes also that with high probability, the
distance between consistent vectors decays like~$O(\frac{N}{M}\log
\frac{M}{\sqrt N})$~(see Corollary~\ref{cor:bound-distance}). 

Similar bounds are also obtained in the case of
QCS of bounded sparse signals when the sensing matrix is a Gaussian
random ensemble. Then, with high probability, if~$M$ exceeds a minimal
number of measurements growing like~$O(\frac{\delta}{\epsilon_0} K
\log \frac{N}{\epsilon_0\sqrt K})$, the distance between two
consistent~$K$-sparse vectors cannot be larger than
$\epsilon_0$  (see Theorem~\ref{thm:consist-impose-closeness-saprse-case}). Equivalently, with high probability, their distance must
then decay like~$O(\frac{K}{M}\log
\frac{MN}{\sqrt K^3})$ (see Corollary~\ref{cor:bound-distance}).

Finally, we evaluate the impact of relaxing the consistency
requirement, \ie allowing for a certain level  of
inconsistent observations, on the reconstruction error of ideal
estimators based on this relaxed condition. This is done by studying
the proximity of two vectors of~$\cl K$ when their mappings in
\eqref{eq:qcs-model} differ by no more than~$r$ components. We show that if this level
$r$ is constant with respect to the evolution of~$M$, \ie what we call
\emph{almost perfect consistency}, the previous
consistency bounds are basically unchanged up to a modification by a multiplicative factor proportional to $(N+r)/N$ for FCQ and to $(K+r)/K$ for QCS (see~Theorem~\ref{thm:bound-distance-almost-consist}). However, in the case where
$r$ can reach a constant fraction of~$M$, \ie $r \leq \rho M$ for~$\rho>0$,
then, provided~$\rho<1/10$, a bias impacts the proximity of two vectors
in~$\cl K$ differing by
no more than~$r$ components in their quantized mapping. With high
probability, the distance
between any such vectors is then smaller than the sum of two terms, one that still
decays with~$M$ and another that is lower bounded
by~$\Omega(\rho\delta)$ (see~Theorem~\ref{thm:propor-consist}).
   
\medskip

The rest of the paper is structured as follows. We start by
providing the precise statement of our main results in
Sec.~\ref{sec:main-results}. In Sec.~\ref{sec:prior-works}, continuing
the literature analysis given above, we
connect our results to a few prior works that are the most connected
to our study. Sec.~\ref{sec:proofs} contains the proofs of
our main results,
while we postpone to Appendix~\ref{sec:proof-lemma} the proof of a key
but more technical lemma, \ie Lemma~\ref{lemma:main-proba-bound}, that sustains
both Theorem~\ref{thm:consist-impose-closeness} and Theorem~\ref{thm:consist-impose-closeness-saprse-case}.

\paragraph*{Conventions:} In the following, we will denote domain
dimensions by capital roman letters, \eg~$M, N$. Vectors and matrices are associated to bold symbols, \eg~$\bs \Phi \in \bb R^{M\times N}$
or~$\bs u \in \bb R^M$, while lowercase light letters
are associated to scalar values. The identity matrix in~$\bb R^D$ reads
$\Id_{D}$. The~$i^{\rm th}$ component of a vector (or of a
vector function)~$\bs u$ reads either~$u_i$ or~$(\bs u)_i$, while the
vector~$\bs u_i$ may refer to the~$i^{\rm th}$ element of a set of
vectors. The set of indices in~$\Rbb^D$ is~$[D]=\{1,\,\cdots,D\}$ and
for any~$\cl S \subset [D]$ of cardinality~$S = \# \cl S$,
$\bs u_{\cl S} \in \bb R^{\# \cl S}$ denotes the restriction of
$\bs u$ to~$\cl S$. For materializing this last operation, we also
introduce the linear restriction operator~$\cl R_{\cl S}$ such that
$\cl R_{\cl S} \bs u = \bs u_{\cl S}$, \ie
$\cl R_{\cl S} = ((\Id_{M})_{\cl S})^T$, where
$\bs B_{\cl S}$ denotes the matrix obtained by restricting the columns
of~$\bs B \in \bb R^{D \times D}$ to those indexed in~$\cl S$.  For
any~$p\geq 1$, the~$\ell_p$-norm of~$\bs u$ is
$\|\bs u\|_p = (\sum_i |u_i|^p)^{1/p}$ with~$\|\!\cdot\!\|=\|\!\cdot\!\|_2$.
The~$(N-1)$-sphere in~$\Rbb^N$ is
$\bb S^{N-1}=\{\bs x\in\Rbb^N: \|\bs x\|=1\}$ while the unit ball is
denoted~$\bb B^{N}=\{\bs x\in\Rbb^N: \|\bs x\|\leq 1\}$. More
generally, we note
$\bb B_s^{N}(\bs q)=\{\bs x\in\Rbb^N: \|\bs x-\bs q\|\leq s\}$. 
We use the simplified notation~$\cl D^{M\times
  N}(\eta)$ and~$\cl D^{M}(\eta)$ to denote an~$M\times N$
random matrix or an~$M$-length random vector, respectively, whose entries are
identically and independently distributed as the probability distribution
$\cl D(\eta)$ of parameters~$\eta = (\eta_1, \cdots, \eta_P)$, \eg
the standard normal distribution~$\cl N(0,1)$ or the uniform
distribution~$\cl U([0, \delta])$. For asymptotic relations, we use the common Landau family of notations,
\ie the symbols~$O$,~$\Omega$ and~$\Theta$~\cite{knuth1976big}. The
positive thresholding function is defined by
$(\lambda)_+ := \tinv{2}(\lambda + |\lambda|)$ for
any~$\lambda\in\Rbb$, and~$\lfloor \lambda \rfloor$ denotes the largest
integer smaller than~$\lambda$.

\section{Main Results}
\label{sec:main-results}

Let us now develop the precise statements of our main results. In this work, we thus focus on the interplay of a uniform midrise quantizer 
\begin{equation}
  \label{eq:quantizer-def}
  \cl Q_\delta(\lambda) := \delta (\lfloor \tfrac{\lambda}{\delta}
  \rfloor + \tinv{2})
  \in \delta (\Zbb + \tinv{2}) =: \bb Z_\delta
\end{equation}
of resolution~$\delta>0$, applied componentwise on vectors, with
both a Gaussian random matrix~$\bs A = \bs \Phi \sim \cl N(0,1)^{M
  \times N}$ and a \emph{dithering}~$\bs \xi \sim \cl U^M([0, \delta])$.
Such a dithering, which must be known at the signal reconstruction, is often used for improving the
statistical properties of the quantizer by randomizing the unquantized
input location inside the quantization cell~\cite{Gray98}. As
will become clear later (see Lemma~\ref{lemma:main-proba-bound} and
its proof in Appendix~\ref{sec:proof-lemma}), this
uniform dithering allows us also to bridge our analysis with a
geometrical probability context inspired by Buffon's needle problem~\cite{buffon_origin,jacques2013quantized}.

Consequently, given a signal~$\bs x$ in a bounded set~$\cl K\subset\bb R^N$, the quantized sensing scenario studied in this paper reads
\begin{equation}
  \label{eq:qcs-unif-dith-model}
  \bs q = \cl Q_\delta[\bs \Phi \bs x + \bs \xi] \in \bb Z^M_\delta.  
\end{equation}
This is either a sensing model for QCS with a Gaussian random sensing~$\bs \Phi$,
or a quantization scheme for
FCQ when the overcomplete frame is made of~$M$ vectors in~$\bb
R^N$ (with~$M \geq N$) that are randomly and independently drawn from~$\cl N(0,
\Id_{N})$. This guarantees that they are also linearly
independent with probability 1, \ie we obtain a \emph{Gaussian random frame}
(GRF) of
$\bb R^N$.

Our main objective in order to quantify the information loss in
\eqref{eq:qcs-unif-dith-model} while trying to estimate~$\bs x$ is to characterize the worst-case error 
$$
\cl E_\delta(\bs \Phi, \bs \xi, \cl K) := \max_{\bs x \in \cl K}\|\bs x -
\bsxalt\|
$$ of any \emph{consistent} reconstruction method whose output~$\bsxalt$ is determined by the following formal program: 
\begin{multline}
  \label{eq:consistent-rec}
  \text{find any}\ \bsxalt\in \bb R^N
 \text{such that}\ \cl Q_\delta(\bs \Phi\bsxalt + \bs \xi) = \cl Q_\delta(\bs \Phi\bs
  x + \bs \xi)\ \text{and}\ \bsxalt \in \cl K.
\end{multline}
In the case of FCQ of signals in a GRF (\ie $M \geq N$), we set~$\cl K = \bb B^N$, while in the context of QCS we take~$\cl K = \Sigma_K(\bs \Psi) \cap \bb B^N$ 
where~$\Sigma_K(\bs \Psi) := \{\bs v =\bs\Psi\bs\alpha\in \bb R^N:
\|\bs \alpha\|_0 \leq K\}$, with~$\|\bs \alpha\|_0 := \#\{j \in [N]:
\alpha_j \neq 0\}$, is the space of~$K$-sparse signals in the
orthonormal basis~$\bs \Psi \in \bb R^{N \times N}$. For the sake of
simplicity, we work with the canonical basis~$\bs \Psi= \Id_{N}$ with~$\Sigma_K := \Sigma_K(\Id_{N})$. However, all our results can be applied to
$\bs \Psi \neq \Id_{N}$ from the rotational invariance of the Gaussian random matrix~$\bs \Phi \sim \cl N^{M\times
  N}(0,1)$ in~$\bb R^N$~\cite{candes2006near}.

We acknowledge the fact that for~$\cl K = \Sigma_K \cap \bb B^N$
the program~\eqref{eq:consistent-rec} is possibly NP
hard, \eg if~$\bsxalt$ is found by minimizing the~$\ell_0$ ``norm''
under the consistency constraint~\cite{natarajan1995sparse}. However, similarly to the procedure developed in
\cite{jacques2013robust}, we are anyway interested in studying its
reconstruction error, remembering that similar
ideal reconstructions in CS and in QCS have often driven the
determination of feasible programs~\cite{candes2006near,tropp2006just,Dai2009,jacques2013robust,plan2011dimension,Jacques2010,Jacques2013}.

Notice that the error~$\cl E_\delta$ is also associated to the biggest size, with
respect to all~$\bs x \in \cl K$, of all consistency
cells~$\cl C_{\bs x} := \{\bsxalt \in \cl K: \cl Q_{\delta}(\bs \Phi \bsxalt + \bs \xi)
= \cl Q_{\delta}(\bs \Phi \bs x + \bs \xi) \}$, \ie 
$$
\cl E_\delta(\bs \Phi, \bs \xi, \cl K) = \max_{\bs x \in \cl K} \max_{\bsxalt \in \cl C_{\bs x}} \|\bs x -
\bsxalt\|,
$$ 
which shows that the characterization of~$\cl E_\delta$ is actually a high
dimensional geometric problem, whose general formulation can be connected to the problem of
finding a finite \emph{covering} of~$\cl K$ of minimal size~\cite{BoufChapter}.
\medskip

Our first contribution is an upper bound on this worst-case error for
consistent signal reconstruction in the context of Gaussian Random Frame
Coefficient Quantization (GRFCQ).
\begin{theorem}[Proximity of consistent vectors -- GRFCQ case]
\label{thm:consist-impose-closeness}
Let us fix~$\epsilon_0 >0$,~$0< \eta < 1$,~$\delta>0$ and~$M\geq N$ such that 
$$
M  \geq \tfrac {4\delta\ +\
  2 \epsilon_0}{\epsilon_0}\,\big(N \log(\tfrac{29 \sqrt{N}}{\epsilon_0}) + \log \tfrac{1}{2\eta}\big). 
$$
Let us randomly draw a GRF~$\bs \Phi \sim \cl N^{M\times N}(0,1)$ and a dithering
$\bs \xi \sim \cl
U^M([0, \delta])$. Then, with probability higher than~$1-\eta$, for all
$\bs x \in \cl K = \bb B^N$ sensed by~\eqref{eq:qcs-unif-dith-model}, any
solution~$\bsxalt$ to~\eqref{eq:consistent-rec} is such that~$\|\bs x - \bsxalt\| \leq \epsilon_0$, or, equivalently,~$\cl E_\delta(\bs \Phi, \bs
\xi, \bb B^N) \leq \epsilon_0$.
\end{theorem}
\begin{proof}
See Sec.~\ref{sec:main-result}.  
\end{proof}

As shown in Sec.~\ref{sec:extension-k-sparse}, it is then
straightforward to adapt Theorem~\ref{thm:consist-impose-closeness}
to quantized observations of sparse signals, \ie to QCS. 
\begin{theorem}[Proximity of consistent vectors -- QCS case]
\label{thm:consist-impose-closeness-saprse-case}
Let us fix~$\epsilon_0 >0$,~$0< \eta < 1$,~$\delta>0$ and~$M$ such that 
$$
M \geq \tfrac {4\delta\ +\ 2 \epsilon_0}{\epsilon_0} \big(2K
\log(\tfrac{56 N}{\sqrt{K}\epsilon_0}) + \log \tfrac{1}{2\eta}\big).
$$
Let us randomly draw a Gaussian sensing matrix~$\bs \Phi \sim \cl
N^{M\times N}(0,1)$ and a dithering~$\bs \xi \sim \cl
U^M([0, \delta])$. Then, with probability higher than~$1-\eta$, for all
$\bs x \in \cl K = \Sigma_K \cap \cl B^N$ sensed by~\eqref{eq:qcs-unif-dith-model}, any
solution~$\bsxalt$ to~\eqref{eq:consistent-rec} is such that~$\|\bs x - \bsxalt\| \leq \epsilon_0$, or, equivalently,~$\cl E_\delta(\bs \Phi, \bs
\xi, \Sigma_K \cap \bb B^N) \leq \epsilon_0$.
\end{theorem}
\begin{proof}
See Sec.~\ref{sec:extension-k-sparse}.  
\end{proof}

As a corollary of those two theorems, the asymptotic decay of~$\cl
E_\delta$ as a function of~$M$,~$N$,~$K$,~$\delta$ and of the
probability~$\eta$ can be established for both GRFCQ and QCS. 
\begin{corollary}[Proximity decay for consistent vectors]
\label{cor:bound-distance}
Given~$M\geq 0$,~$0< \eta < 1$,~$\delta>0$ with~$\delta =
O(1)$, there exists two constants~${C,C'>0}$ such that
\begin{equation}
  \label{eq:worst-case-error-GRFCQ}
  \bb P\big[\cl E_\delta(\bs\Phi, \bs \xi,\bb B^N) \leq C \big(
  \tfrac{N}{M}\log\tfrac{M}{\sqrt N} + \log\tfrac{1}{2\eta} \big) \big]
  \geq 1 - \eta,
\end{equation}
for GRFCQ, and 
\begin{equation}
  \label{eq:worst-case-error-QCS}
  \bb P\big[\cl E_\delta(\bs \Phi, \bs \xi, \Sigma_K \cap \bb B^N) \leq C' \big(\tfrac{K}{M} \log(\tfrac{MN}{\sqrt{K^3}}) + \log\tfrac{1}{2\eta} \big) \big]
  \geq 1 - \eta,
\end{equation}
for QCS, where these probabilities are computed with respect to both~$\bs \Phi \sim \cl N^{M\times
  N}(0,1)$ and the dithering~$\bs \xi \sim \cl
U^M([0, \delta])$.
\end{corollary}
\begin{proof}
See Sec.~\ref{sec:main-result} for GRFCQ and
Sec.~\ref{sec:extension-k-sparse} for QCS.
\end{proof}

Loosely speaking, this corollary says that if~$\delta = O(1)$, with probability exceeding~$1-\eta$, 
\begin{equation}
  \label{eq:fcq-bound}
  \cl E_\delta(\bs \Phi, \bs \xi, \bb B^N) =
  O(\tfrac{N}{M}\log\tfrac{M}{\sqrt N} + \log\tfrac{1}{2\eta}),  
\end{equation}
for GRFCQ, and
\begin{equation}
  \label{eq:worst-case-error-sparse-case}
  \cl E_\delta(\bs \Phi, \bs \xi, \Sigma_K \cap \bb B^N) = O(\tfrac{K}{M} \log(\tfrac{MN}{\sqrt{K^3}}) + \log\tfrac{1}{2\eta}),   
\end{equation}
for QCS. As explained in Sec.~\ref{sec:prior-works}, this matches existing
error bounds for 1-bit compressed sensing in the case of Gaussian
random projections~\cite{jacques2013robust}. It also improves upon previous
known bounds, decaying as~$O(1/{\sqrt M})$ for linear reconstruction methods
in FCQ~\cite{goyal_1998_lowerbound_qc} and as~$O(\sqrt{{K}/{M}})$ for QCS, while a
known lower bound in~$\Omega({K}/{M})$ exists~\cite{BoufChapter}.  Our
result behaves also similarly to the bound on the mean worst-case
error (established with respect to~$\bs \Phi$ and~$\bs \xi$ in
the context of our notation) of consistent reconstruction methods obtained in
\cite{powell_consistent} in the case of random frames over~$\bb S^{N-1}$.

Our last contribution shows that small deviations to strict
consistency are also possible while keeping control of the proximity between
almost-consistent vectors. This allows us to consider a moderate 
corruption of the sensing model~\eqref{eq:qcs-unif-dith-model}, \eg in the case where it suffers from a prequantization noise~$\bs n \in \bb R^M$ such that, for any
$\bs x \in \cl K$, the bound
\begin{equation}
  \label{eq:corrupted-QCS-model}
  \|\cl Q_\delta((\bs \Phi\bs x + \bs n) + \bs \xi) - \cl Q_\delta(\bs
  \Phi\bs x + \bs \xi)\|_1 \leq s\,\delta 
\end{equation}
holds with high probability for some~$s>0$.  In such a case, we can
relax the formal reconstruction~\eqref{eq:consistent-rec} and attempt
to reconstruct~$\bsxalt$ from the program:
\begin{multline}
  \label{eq:consistent-rec-robust}
  \text{find any}\ \bsxalt\in \bb R^N\ \text{such that}\\ \|\cl Q_\delta(\bs \Phi\bsxalt + \bs \xi) - \cl Q_\delta((\bs \Phi\bs
  x + \bs n ) + \bs \xi)\|_1\leq s\,\delta\\ \text{and}\ \bsxalt \in \cl K,
\end{multline}
By construction, we observe that the solution~$\bsxalt \in \cl
  K$ is such that 
  \begin{align*}
&\|\cl Q_\delta(\bs \Phi\bsxalt + \bs \xi) - \cl Q_\delta(\bs \Phi\bs
  x  + \bs \xi)\|_1\\
&\leq \|\cl Q_\delta(\bs \Phi\bsxalt + \bs \xi) - \cl Q_\delta((\bs \Phi\bs
  x + \bs n ) + \bs \xi)\|_1\\
&\quad + \|\cl Q_\delta((\bs \Phi\bs
  x + \bs n ) + \bs \xi) - \cl Q_\delta(\bs \Phi\bs x + \bs \xi)\|_1
  \leq 2 s \delta.
\end{align*}
Therefore, characterizing the robustness of~\eqref{eq:consistent-rec-robust}
amounts to studying the proximity of vectors having approximately consistent
quantized random projections, \ie we have to analyse the ``largest'' relaxed consistency cell
$\cl C^r_{\bs x} := \{\bsxalt \in \cl K: \|\cl Q_{\delta}(\bs \Phi \bsxalt + \bs \xi)
- \cl Q_{\delta}(\bs \Phi \bs x + \bs \xi)\|_1 \leq r\,\delta \}$ for
$r = 2s >0$, through the worst-case error  
$$
\cl E^r_\delta(\bs \Phi, \bs \xi, \cl K) := \max_{\bs x \in \cl K} \max_{\bsxalt \in \cl C^r_{\bs x}} \|\bs x -
\bsxalt\|.
$$ 

As shown in the next theorem, this is done first in the case
where~$r=O(1)$ relatively to~$M$, \ie for ``almost perfect consistency''.
\begin{theorem}[Proximity decay for almost perfectly consistent vectors]
\label{thm:bound-distance-almost-consist}
Given~$M\geq 0$,~$0< \eta < 1$,~$\delta>0$ with~$\delta =
O(1)$, and~$r\in \bb N$ with~$r = O(1)$, there exists two constants~$C,C'>0$ such that
\begin{equation}
  \label{eq:worst-case-error-GRFCQ-almost-consist}
  \hspace{-3mm}\bb P\big[\cl E^r_\delta(\bs\Phi, \bs \xi,\bb B^N) \leq C \tfrac
  {N+r}{M} \big(\log(\tfrac{M \max(N,M)}{N}) +
\log \tfrac{1}{2\eta} \big) \big]  \geq 1 - \eta,
\end{equation}
for GRFCQ, and 
\begin{equation}
  \label{eq:worst-case-error-sparse-case-almost-consist}
  \bb P\big[\cl E^r_\delta(\bs \Phi, \bs \xi, \Sigma_K \cap \bb B^N)
  \leq C' \tfrac {K+r}{M} \big(\log(\tfrac{M \max(N,M)}{K}) +
\log \tfrac{1}{2\eta} \big) \big] 
  \geq 1 - \eta,
\end{equation}
for QCS, where these probabilities are computed with respect to both~$\bs \Phi \sim \cl N^{M\times
  N}(0,1)$ and the dithering~$\bs \xi \sim \cl
U^M([0, \delta])$.
\end{theorem}
\begin{proof}
See Sec.~\ref{sec:a-almost-perfect}.
\end{proof}

This theorem points out that, in \emph{almost perfect consistent}
reconstruction regime, \eg if the variance of the prequantization noise
components~$n_j$ rapidly decreases with~$j$ in~\eqref{eq:corrupted-QCS-model} the impact of the noise~$\bs n$ on the
reconstruction of~$\bs x$ is controlled and does not change the
asymptotic decay of the worst-case reconstruction error. Loosely
speaking,~$\cl E^r_\delta$ behaves like
$\cl E_\delta$ up to a multiplication by
$(N+r)/N$ for GRFCQ and by~$(K+r)/K$ for QCS. 

However, assuming as above that~$\bs n$ vanishes with its component index is a rather rare scenario as noise is often considered as a
stationary phenomenon, \ie it is more reasonable to assume an equal
probability of corruption on all observations. In order to address
this more realistic situation, we show that in the case of
\emph{proportional inconsistency} where we only know that~$r \leq \rho M$ for some
constant~$0<\rho<1$, the worst-case reconstruction error suffers from
a systematic bias induced by our ignorance of the indices of the
corrupted quantized projections. This situation occurs for instance if~\eqref
{eq:corrupted-QCS-model} is corrupted by a homoscedastic, zero-mean
noise~$\bs n \in \bb R^M$, \ie ${\rm Var}(n_i) = \sigma^2$ for~$i\in [M]$ and
$\sigma > 0$. More specifically, if~$\bs n \sim \cl N(0,\sigma^2
\Id_M)$, using the law of total expectation sequentially over the
  dithering and on the noise, we have that~$\bb E \|\cl Q_\delta(\bs\Phi\bs x + \bs n + \bs \xi) -
  Q_\delta(\bs\Phi\bs x + \bs \xi)\|_1 = \sqrt{2/\pi}\,\sigma
  M$, \ie $s$ tends to~$\sqrt{2/\pi}(\sigma/\delta)
  M$ for large~$M$.

In particular, we demonstrate the following result. 
\begin{theorem}[Proximity decay for proportionally inconsistent vectors]
\label{thm:propor-consist}
Let~$0<\rho<1$ be such that 
\begin{equation}
  \label{eq:cond-rho}
  \bar\rho := \rho\,(1 + 2\log(e/\rho)) < 1.  
\end{equation}
If
\begin{equation}
  \label{eq:conc-M-almost-consist-GRFCQ}
\ts M \geq \tfrac{4\delta + 4}{\epsilon_0}\big(N \log \frac{29 \sqrt  N}{\epsilon_0} + \log \frac{1}{2\eta}\big),
\end{equation}
for GRFCQ and~$\cl K = \bb R^N$, or if 
\begin{equation}
  \label{eq:conc-M-almost-consist-QCS}
M \geq \tfrac {4\delta + 2}{\epsilon_0} \big( 2K \log(\tfrac{56
  N}{\sqrt{K}\epsilon_0}) + \log \tfrac{1}{2\eta} \big),  
\end{equation}
for QCS and~$\cl K = \Sigma_K \cap \bb B^N$, then, with probability at
least~$1-\eta$,  
\begin{equation}
\label{eq:thm-almost-consist}
\forall \bs x,\bsxalt\in \cl K, \| \cl Q_\delta(\bs \Phi \bs x + \bs
\xi)-\cl Q_\delta(\bs \Phi \bsxalt + \bs \xi)\|_1 \leq \rho\delta M\
\Rightarrow\quad\|\bs x - \bsxalt\|\leq C_\rho \epsilon_0 + D_\rho \delta,
\end{equation}
with~${C_\rho := \frac{1}{1-\bar\rho} \geq 1}$ and~$D_\rho :=
4 \rho\, C_\rho \log(e/\rho) \geq 4
\rho$.   
\end{theorem}
\begin{proof}
See Sec.~\ref{sec:prop-inconsist}.    
\end{proof}

In the error bound~$C_\rho \epsilon_0 + D_\rho \delta$, the announced bias is
materialized by the constant second term~$D_\rho \delta \geq 4 \rho
\delta$. Conversely to the term~$C_\rho \epsilon_0$ that can be made
arbitrarily small by reducing~$\epsilon_0$ and increasing~$M$, this
bias limits our capability to approach~$\bs x$ with~$\bsxalt$, or
equivalently to
estimate~$\bs x$ from its corrupted quantized observation in the
reconstruction program~\eqref{eq:consistent-rec-robust}.

\begin{remark}
Notice that the condition~\eqref{eq:cond-rho}
holds if~$\rho \leq 1/10$. Moreover, since for~$\rho = 1/10$,~$C_\rho < 4.2~$ and~$D_\rho < 1.7$, and
since both~$C_\rho$ and~$D_\rho$ are convex and non-decreasing over~$\rho \in (0,1/10]$,
we get more simply over this interval
$$
\|\bs x - \bsxalt\|\ \leq\ C_\rho \epsilon_0 + D_\rho \delta\ \leq\ 4.2 \epsilon_0 + 17 \rho\delta,
$$
under the same conditions as above. 
\end{remark}

\begin{remark}
The conditions \eqref{eq:conc-M-almost-consist-GRFCQ}
and \eqref{eq:conc-M-almost-consist-QCS} are basically unchanged compared to
those imposed on~$M$ in Theorems~\ref{thm:consist-impose-closeness} and
\ref{thm:consist-impose-closeness-saprse-case}, respectively. Therefore, the
reader can easily show that the reasoning providing
Corollary~\ref{cor:bound-distance} applies here if we saturate the
conditions on~$M$ in Theorem~\ref{thm:propor-consist} in order to study the decay of $\epsilon_0$. In
particular, considering the first remark, for $r\leq\rho M$
with $\rho<1/10$ and $\delta =
O(1)$ relatively to $M$, we have with probability exceeding~$1-\eta$,
$$  
\cl E^r_\delta(\bs \Phi, \bs \xi, \bb B^N) = O(\rho \delta + \tfrac{N}{M}\log\tfrac{M}{\sqrt N} + \log\tfrac{1}{2\eta}),  
$$
for GRFCQ, and
$$
  \cl E^r_\delta(\bs \Phi, \bs \xi, \Sigma_K \cap \bb B^N) =  O(\rho
  \delta + \tfrac{K}{M} \log(\tfrac{MN}{\sqrt{K^3}}) + \log\tfrac{1}{2\eta}),   
$$
for QCS.
\end{remark}

\begin{remark}
\label{rem:tight-bias}
Notice finally that the second term $D_\rho \delta$ in \eqref{eq:thm-almost-consist}
representing the constant bias induced by the proportional
inconsistency of $\bs x$ and $\bsxalt$ is actually
\emph{necessary}. In the case of QCS, taking a
$K$-sparse vector $\bs x \in \cl K := \Sigma_K \cap \bb B^N$, for $\bsxalt = \bs x +
\lambda \delta \bs e_i$ with $i \in \supp \bs x$ and $\lambda \in \bb R$ such that $\bsxalt \in
\cl K$, it is easy to see that for large
$M$, by the law of large numbers,  
\begin{equation*}
\ts \| \cl Q_\delta(\bs \Phi \bs x + \bs
\xi)-\cl Q_\delta(\bs \Phi \bsxalt + \bs \xi)\|_1\
= \sum_i | \cl Q_\delta(z_i)-\cl Q_\delta(z_i + \lambda \delta
\Phi_{i1})|\  
\approx\ M \delta\,\bb E X,
\end{equation*}
with $\bs z := \bs \Phi \bs x + \bs
\xi$ and $X := \#\{\delta \bb Z \cap [u, u +
\lambda \delta g]\}$ with $u \sim \cl U([0,\delta])$ and
$g\sim \cl N(0,1)$, \ie $X$ is the discrete random variable counting the
elements of $\delta \bb Z$ falling in $[u, u +
\lambda \delta g]$.  We can then determine that $\bb E X = \bb E_g \bb E_u(X|g) = \bb E_g
\tfrac{|\lambda| \delta |g|}{\delta} = \tfrac{2}{\pi} |\lambda|$ (see,
\eg \cite{jacques2013quantized}), so
that 
$$
\| \cl Q_\delta(\bs \Phi \bs x + \bs
\xi)-\cl Q_\delta(\bs \Phi \bsxalt + \bs \xi)\|_1 \approx
\tfrac{2}{\pi} |\lambda| \delta\, M \leq \rho \delta\,M.
$$
In other words, for large values of $M$, one can find two vectors $\bs x$ and
$\bsxalt$ satisfying $\| \cl Q_\delta(\bs \Phi \bs x + \bs
\xi)-\cl Q_\delta(\bs \Phi \bsxalt + \bs \xi)\|_1\leq \rho \delta M$ with $\rho \propto |\lambda|$,
while $\|\bs x - \bsxalt\| = |\lambda|\delta \gtrsim
\rho \delta$. 

A similar explanation could simply use the quasi-isometric property
of the embedding $\bs x \to \cl Q_\delta(\bs \Phi \bs x + \bs
\xi)$ studied in the Quantized Johnson-Lindenstrauss lemma of
\cite[Prop. 14]{jacques2013quantized}, to show that, on the same vector pair $(\bs
x,\bsxalt)$ and with high probability, $\|\bs x - \bsxalt\|
\geq c \rho \delta$ if $\| \cl Q_\delta(\bs \Phi \bs x + \bs
\xi)-\cl Q_\delta(\bs \Phi \bsxalt + \bs \xi)\|_1 \geq c'\delta\rho\,
M$, for two universal constants $c,c'>0$.
\end{remark}

\section{Discussion}
\label{sec:prior-works}

Recently, Powell and Whitehouse in~\cite{powell_consistent} have analyzed a model
equivalent to~\eqref{eq:qcs-unif-dith-model} by adopting a geometric standpoint. In
particular, adapting their work to our notations, they have studied
the sensing model
$$
\bs q = \bs A \bs x + \bs n,
$$  
where~$\bs A = (\bs a^T_1, \cdots, \bs a^T_M)^T\in \bb R^{M \times N}$
is a frame whose elements~$\bs a_j$ are drawn from a suitable
distribution on
$\bb S^{N-1}$ and the uniform noise~$\bs n \sim \cl U^{M}([-\delta, \delta])$
stands for, \eg a dithered uniform scalar quantization of~$\bs A \bs x$.
They observe that the consistent reconstruction polytope, which has at
most~$2M$ faces of dimension~$N-1$, 
$$
Q_M := \{\bs u \in \bb R^N: \|\bs A \bs u - \bs q\|_\infty \leq \delta\}
$$
can be seen a translation of an \emph{error polytope}~$P_M$, \ie for any
consistent reconstruction~$\bsxalt \in Q_M$
$$
(\bsxalt - \bs x)\ \in\ P_M\ :=\ \{\bs u \in \bb R^N: \|\bs A \bs u - \bs n\|_{\infty}
\leq \delta\}.
$$
Therefore, for a given~$\bs A$, analyzing the worst-case error of any
consistent reconstruction amounts to estimating the width of~$P_M$,
\ie
\begin{align*}
W_M&= \max\{\|\bs u\|:\bs u \in P_M\}\\
&= \max\{\|\bs u - \bs x\|: \bs u
\in Q_M\}.  
\end{align*}
Authors in~\cite{powell_consistent} estimate the expected worst-case square error~$\bb E |W_M|^2$ with
respect to the distribution of the random vectors~$\{\bs a_j: 1\leq j
\leq M\}$ on~$\bb S^{N-1}$. Relating this estimation to coverage
processes on the unit sphere~\cite{burgisser2010coverage}, they show that, under general assumption on the
distribution of these unit frame vectors, 
$$
(\bb E |W_M|^2)^{\inv{2}} \leq \tfrac{C\delta}{M},
$$
with~$C>0$ depending on this distribution. In particular, for~$M$ frame
vectors uniformly drawn at random over~$\bb S^{N-1}$,~$C = O(N^{3/2})$ so that 
\begin{equation}
  \label{eq:powell-bound}
  (\bb E |W_M|^2)^{\inv{2}} = O(\tfrac{N^{3/2} \delta}{M}).  
\end{equation}

Despite a slightly different context where the results above focus on
an expected worst-case analysis, the behavior of these bounds is
highly similar to the one we get in Corollary~\ref{cor:bound-distance}
for consistent reconstruction of signals in the case of GRFCQ: we observe
that, for one draw of this ($M/N$)-redundant GRF and of the
quantization dithering,~$\cl E_\delta = O(\tfrac{N}{M}(\log\tfrac{M}{\sqrt N}+\log\tinv{2\eta}))$ with
probability higher than~$1-\eta$. 

At first sight the dependence in~$N^{3/2}$ of~\eqref{eq:powell-bound}
may seem less optimal than the dependence in~$N$ of
\eqref{eq:fcq-bound}. However, the first bound is adjusted to random
frame vectors drawn uniformly at random over~$\bb S^{N-1}$~\cite{powell_consistent}, \ie they
all have a unit norm while the GRF vectors have expected
length equal to~$\sqrt{N}$. Keeping in mind the difficulty to
compare a bound on the expectation of a random event with a
probabilistic bound on this event itself, we can notice, however, that rescaling the result of
\cite{powell_consistent} to uniform random frames over the dilated
sphere~$\sqrt{N}\,\bb
S^{N-1}$, or conversely rescaling~$\delta$ into~$\delta/\sqrt{N}$
in~\eqref{eq:powell-bound}, provides an error decay in~$(\bb E |W_M|^2)^{\inv{2}} = O(\tfrac{N\delta}{M})$.   

The reader can notice that the decay in~$(\log M)/M$ of our bound~\eqref{eq:fcq-bound}
with respect to~$M$ suffers from an extra log factor compared to the decay of~$(\bb
E |W_M|^2)^{1/2}$ in~\eqref{eq:powell-bound}. Actually, the
  same observation can be made with respect to known bounds obtained in
  the prior works summarized in
  the Introduction. These former studies have indeed focused on
  characterizing upper bounds on the \emph{mean} square error (MSE) of
  (almost) consistent signal estimation when frame coefficients are
  quantized or, equivalently, when they are corrupted by uniform
noise. In their settings, the signal is
assumed \emph{fixed}, the corrupting quantization
noise is random, and the frame construction is either
random~\cite{rangan_rec_algo_2001,powell_RGalgo} or deterministic~\cite{powell_RGalgo}.
They evaluated the MSE of (almost) consistent estimates over the sources of
randomness and basically proved them to be bounded as~$\bb E\|\bs x
-\bsxalt\| = O(N/M)$. This was also sustained by
empirical evidences in
\cite{Thao_Consistency_1994,goyal_1998_lowerbound_qc}, and aligned
with lower bounds in~$\Omega(N/M)$ on (Bayesian) MSE generally evaluated on
random signal construction~\cite{goyal_1998_lowerbound_qc,rangan_rec_algo_2001}.
The only exception to this averaged setting comes from~\cite{goyal_1998_lowerbound_qc} which, in the particular case of
a tight frame formed by an oversampled Discrete Fourier Transform (DFT), proves that a consistent
estimate of~$\bs x$ reached a \emph{squared error} bounded as~$\|\bs x
- \bsxalt\|^2 = O(N^2/M^2)$ under a mild assumption on~$\bs x$. They
also conjectured that the MSE for any~$(M/N)$-redundant frame should
decay as~$O(N^2/M^2)$.

As will become clear in Sec.~\ref{sec:proofs}, the source of the
extra log factor in~\eqref{eq:fcq-bound} (and similarly in
\eqref{eq:worst-case-error-sparse-case} for QCS of sparse signals) is
due to our implicit worst-case error analysis of (almost) consistent signal estimations. By construction, this means that, with high probability on the draw of the
matrix~$\bs \Phi$ and on the dithering, our results are valid
uniformly for all vectors of the bounded set~$\cl K$, \ie for~$\cl K=\bb
B^N$ in the case of GRFCQ or for~$\cl K=\Sigma_K\cap\bb B^N$ in the case of QCS. Practically, these log factors are induced by the
use of union
bound arguments in the proofs of
Theorems~\ref{thm:consist-impose-closeness} and
\ref{thm:consist-impose-closeness-saprse-case} for upper bounding the probability of failure of our
error bounds over all elements of a covering set of~$\cl K$
(see Sec.~\ref{sec:proofs}).

There exist also other works in QCS interested in asymptotic
  regimes where the three dimensions~$(K,M,N)$ become arbitrary large,
  \eg keeping~$M/K$ and/or~$N/M$ constant, and where the signal is generated randomly from a continuous
  distribution (see \eg~\cite{goyal_2008_cs_lossycomp}). At first
  sight, such
  approaches seem incompatible with the boundedness of the signal
  domain~$\cl K$ assumed in this work, \eg with~$\cl K \subset \bb B^N$. Indeed, as in~\cite{goyal_2008_cs_lossycomp,kamilov_2012}, 
  if the input signal~$\bs x$ is random with i.i.d. entries distributed as a
  distribution~$p_x$, one that ensures the signal to be
  (approximately)~$K$-sparse (\eg with a Gauss-Bernoulli distribution), the expected signal norm~$\bb E\|\bs
  x\|$ is not bounded and grows like~$\sqrt K$. In such a context, a Gaussian
  random sensing matrix must be generated as~$\bs\Phi \sim \cl N^{M\times N}(0,
  1/K)$ in order to keep a constant variance for the components of
 ~$\bs\Phi \bs x$~\cite{goyal_2008_cs_lossycomp}, and hence
  maintaining a constant resolution for the quantizer
  whatever the configuration of~$(K,M,N)$. Provided that~$M$ grows
  like~$O(K \log N/K)$ so that the matrix~$\sqrt{K/M}\,\bs \Phi$ is RIP with
  high probability, one can estimate any signal from the QCS model~\eqref{eq:qcs-unif-dith-model} using, \eg the BPDN
  program~\cite{Chen98atomic,candes2006ssr}. Then, considering the scaling of the sensing matrix
  entries\footnote{If~$\tinv{\mu} \bs \Phi$ is RIP of order~$2K$ and
    constant~$\delta < \sqrt
    2 - 1$, then, the BPDN reconstruction
  error obtained on the noisy sensing model~$\bs y = \bs \Phi \bs x + \bs n$ with~$\|\bs n\| \leq \epsilon$ and~$\bs x$ sparse is bounded by
 ~$O(\epsilon/\mu)$~\cite{Jacques2010}.}, the signal reconstruction
  error~$\|\bs x - \bsxalt\|$ obtained from BPDN is bounded by~$O(\sqrt
  K\,\delta)$, so that~$\|\bs x - \bsxalt\|/\|\bs x\| =
  O(\delta)$. This illustrates again the limit of reconstruction methods
  that do not promote consistency as this error bound is not
  decaying when~$M$ increases. Actually, a recent consistent reconstruction
  method for QCS based on a message passing algorithm~\cite{kamilov_2012} observes
  empirically that~$\|\bs x - \bsxalt\|$ decays as~$1/M$.

  The interested reader will easily check that the present work
  can be adapted to such unbounded signal sensing. Indeed, in the QCS model
 ~\eqref{eq:qcs-unif-dith-model}, we can always apply the variable
  changes~$\bs \Phi = \sqrt K\,\bs \Phi'$ and~$\bs x = \bs x'/\sqrt K$
  with~$\bs \Phi \sim \cl N^{M \times N}(0,1)$, so that~$\bs \Phi
  \bs x = \bs \Phi' \bs x'$. Moreover,~$\bs x \in \cl K =
  \Sigma_K \cap \bb B^N$ involves that~$\bs x' \in \cl K' := \sqrt K\,\cl
  K = \Sigma_K \cap \sqrt K\,\bb B^N$. Therefore, as we basically show in
  Sec.~\ref{sec:main-results} that, with high probability, any consistent and sparse signal estimate
 ~$\bsxalt$ of~$\bs x$ respects~$\|\bs x - \bsxalt\|= O(\frac{K}{M}\log
\frac{MN}{K^{3/2}})$, this establishes that, using the
sensing/signal domain combination~$(\bs\Phi',\cl K')$, a sparse and consistent
estimate~$\tilde{\bs x}$ of~$\bs x' \in \cl K'$ necessarily satisfies~${\|\bs x'
  - \tilde{\bs x}\|}/{\|\bs x'\|} = O(\frac{K}{M}\log
\frac{MN}{K^{3/2}})$ under the same conditions.  Up to the extra log
factor already discussed above, this meets the most recent empirical
observations made in~\cite{kamilov_2012}. Moreover, we observe quickly
that, conversely to BPDN and to equivalent reconstruction approaches, the reconstruction error of
consistent signal estimation vanishes if~$M > K^{1+c}$ and~$M > d N$
for any~$c,d>0$ while~$K$,~$M$ and~$N$
tend all to infinity.

To conclude this section, as pointed out by the known lower bounds
described in the Introduction, let us mention that regular scalar
quantization provides a rather limited decay of the reconstruction
error, both for FCQ and
QCS contexts.  Recent developments in vector quantization for FCQ
\cite{vivekQuantFrame}, in the use of feedback quantization and of
the~$\Sigma\Delta$ scheme for FCQ~\cite{KSW12} and QCS
\cite{gunturk2013sobolev,BoufChapter}, and finally non-regular
quantization schemes where~$\cl Q$ is periodic over its range~\cite{B_TIT_12,BR_DCC13,pai_nonadapt_MIT06,kamilov_2012}, provide all faster reconstruction error
bounds decaying polynomially or even exponentially in~$M$. The
implicit objective of this paper is therefore to improve our
understanding of one of the simplest quantization schemes, that is basically a
dithered round-off operation and its combination with Gaussian random
projections.   

\section{Proofs}
\label{sec:proofs}
\subsection{Quantization of Gaussian Random Frame Coefficients}
\label{sec:main-result}

This section is dedicated to proving
Theorem~\ref{thm:consist-impose-closeness} and the GRFCQ part of
Corollary~\ref{cor:bound-distance}.  Following an argument developed in~\cite{B_TIT_12} for
non-regular scalar quantization, proving that
\begin{equation}
  \label{eq:bounds-max-error-GRF}
  \cl E_\delta(\bs \Phi, \bs \xi, \bb B^N) = \max_{\bs x \in \bb B^N}
  \max_{\bsxalt \in \cl C_{\bs x}}\, \|\bs x -
  \bsxalt\|\quad \leq\quad \epsilon_0  
\end{equation}
holds with probability exceeding~$1-\eta$ on the draw of a GRF~$\bs\Phi = (\bs \varphi_1,\,\cdots,\bs
\varphi_M)^T \sim \cl N^{M \times N}(0,1)$ and of a \emph{dithering}~$\bs
\xi \sim \cl U^{M}([0,\delta])$, amounts to showing that
\begin{equation*}
\bb P[\forall \bs x,\bsxalt \in \bb B^N, \cl Q_\delta[\bs \Phi \bs x + \bs \xi] = \cl Q_\delta[\bs \Phi \bsxalt
+ \bs \xi]\ \Rightarrow\ \|\bs x - \bsxalt\| \leq \epsilon_0\big] \geq
1 - \eta,
\end{equation*}
where~$\bb P$ is computed with respect to the random quantities~$\bs
\Phi$ and~$\bs \xi$.  

Taking the contraposition, we can alternatively demonstrate that, 
\begin{equation*}
P_{\rm fail} := \bb P\big[\exists\, \bs x, \bsxalt \in \bb B^N, \|\bs x - \bsxalt\| \geq
\epsilon_0\ \text{s.t.}\ Q_\delta[\bs \Phi \bs x + \bs \xi] = Q_\delta[\bs \Phi \bsxalt + \bs \xi]\big] \leq \eta.  
\end{equation*}

For upper bounding~$P_{\rm fail}$, we take an~$s$-\emph{covering} of the
unit ball~$\bb B^{N}$, \ie a finite point set~$\cl L_s$ such that for any~$\bs v\in \bb B^N$, there exists a point
$\bar{\bs v} \in \cl L_s$ with distance at most~$s$ from
$\bs v$, \ie $\|\bs v - \bar{\bs v}\| \leq s$. The cardinality~$L_s=\#\cl L_s$ of this
covering set is known to be bounded as~$L_s \leq (3/s)^N$~\cite{baraniuk2008simple}.  

Therefore, if~$\bs x,\bsxalt \in \bb B^{N}$ are such that~$\|\bs x - \bsxalt\| \geq \epsilon_0$, taking their respective closest points
$\bar{\bs x}, \bar{\bs x}^* \in \cl L_s$, we have
$
\|\bar{\bs x} -\bar{\bs x}^*\| \geq \epsilon_0 - 2s
$. 
Consequently, it follows that 
\begin{multline*}
P_{\rm fail}\leq\ \bb P\big[(\exists \bar{\bs p}, \bar{\bs q}\in \cl
L_{s}: \|\bar{\bs p}-\bar{\bs q}\| \geq
\epsilon_0 - 2s),\\ \exists \bs u \in \bb B_{s}(\bar{\bs p}), \exists \bs v \in
\bb B_{s}(\bar{\bs p}):\ \cl Q_\delta[\bs \Phi \bs u + \bs \xi] =\cl Q_\delta[\bs \Phi \bs v + \bs \xi]\,\big].
\end{multline*}
Indeed, if the event whose probability is measured by~$P_{\rm fail}$ is
verified for~$\bs x$ and~$\bsxalt$, taking~$\bar{\bs p}=\bar{\bs x}$,~$\bar{\bs q} = \bar{\bs x}^*$,
$\bs u=\bs x$ and~$\bs
v=\bsxalt$ shows that the event associated to the probability of the RHS
above occurs.

Thus, if one can find an upper bound~$P_0$ on 
\begin{equation*}
\bb P\big[\exists \bs u \in \bb B_{s}(\bar{\bs p}), \exists \bs v \in
\bb B_{s}(\bar{\bs p}),\ \cl Q_\delta[\bs \Phi \bs u + \bs \xi] =\cl
Q_\delta[\bs \Phi \bs v + \bs \xi]
\big|\ \|\bar{\bs p}-\bar{\bs q}\| \geq
\epsilon_0 - 2s\big] \leq P_0  
\end{equation*}
that is independent of~$\bar{\bs p}$ and~$\bar{\bs q}$, since the
number of possible pairs of points in~$\cl
L_{s}$ is bounded by~${L_s \choose 2} < \tinv{2} L^2_{s}$ independently of any conditions on them, a union
bound provides
$$
P_{\rm fail} \leq \tinv{2} L^2_{s} P_0.
$$
The following key lemma allows one to estimate~$P_0$.
\begin{lemma} 
\label{lemma:main-proba-bound}
Let~$\tilde{\bs p}, \tilde{\bs
  q}$ be two points in~$\bb R^N$. There exists a radius~$s' \geq \tfrac{1}{8\sqrt N} \|\tilde{\bs p}-\tilde{\bs
  q}\|$ such that, for~$\bs \Phi \sim \cl N^{N \times
  M}(0,1)$ and~$\bs \xi \sim \cl U^{M\times 1}([0,\delta])$, the probability  
\begin{equation*}
P_{s'}(\alpha, M)\ :=\ \bb P\big[\exists \bs u \in \bb B_{s'}(\tilde{\bs p}), \exists \bs v \in
\bb B_{s'}(\tilde{\bs q}),\ \cl Q_\delta[\bs \Phi \bs u + \bs \xi] =\cl
Q_\delta[\bs \Phi \bs v + \bs \xi]\big]  
\end{equation*}
satisfies
\begin{equation}
  \label{eq:bound-P-sp}
  P_{s'}(\alpha, M)\ \leq\ \big(1 - \tfrac{3\alpha}{8\ +\ 4\alpha}\big)^M,  
\end{equation}
with~$\alpha = {\|\tilde{\bs p}-\tilde{\bs q}\|}/{\delta}$.
\end{lemma}
\begin{proof}
See Appendix~\ref{sec:proof-lemma}.  
\end{proof}

As explained in its proof (see
Appendix~\ref{sec:proof-lemma}), this lemma is determined by an equivalence with Buffon's Needle problem
in~$N$ dimensions~\cite{hey2010georges}, where the needle is actually replaced by a ``dumbbell'' shape
whose two balls are associated to the two neighborhoods of~$\tilde{\bs
  p}$ and~$\tilde{\bs q}$. 

The quantity~$P_{\lambda}(\alpha, M)$ defined in Lemma~\ref{lemma:main-proba-bound} increases with
$\lambda>0$. Therefore, for finding an estimate of~$P_0$ which is
associated to the covering radius~$s$, we must guarantee that~$s \leq s'$, knowing that~$\|\bar{\bs p}-\bar{\bs q}\| \geq \epsilon_0 - 2s$ and~$2s' \geq \tinv{4\sqrt N} \|\bar{\bs p} - \bar{\bs q}\|$.  This is achieved by imposing~$\tinv{4\sqrt N} (\epsilon_0 - 2s) = 2s$, \ie 
$$
2s = \tfrac{\epsilon_0}{4\sqrt N + 1}.  
$$
This provides also~$\epsilon_0 - 2s = \tfrac{4\sqrt N}{4\sqrt N +
  1}\,\epsilon_0 > \tfrac{4}{5}\epsilon_0$ if~$N\geq 2$.

Consequently, using~\eqref{eq:bound-P-sp} and observing that~$1 -
\,3\alpha/(4 + 8\alpha)$ decays with~$\alpha$, 
\begin{align*}
&\bb P\big[\exists \bs u \in \bb B_{s}(\bar{\bs p}), \exists \bs v \in
\bb B_{s}(\bar{\bs p}),\\
&\qquad\cl Q_\delta[\bs \Phi \bs u + \bs \xi] =\cl
Q_\delta[\bs \Phi \bs v + \bs \xi]\,\big|\ \|\bar{\bs p}-\bar{\bs q}\|
\geq
\epsilon_0 - 2s\big]\\
&= P_{s}\big(\tinv{\delta}\|\bar{\bs p}-\bar{\bs q}\|, M\big)\\
&\leq P_{s'}\big(\tinv{\delta}\|\bar{\bs p}-\bar{\bs q}\|, M\big)\\
&\leq \big(1 - \tfrac{3\frac{4}{5}\epsilon_0}{8\delta\ +\
  4\frac{4}{5}\epsilon_0}\big)^M\ <\ \big(1 -
  \tfrac{2\epsilon_0}{8\delta\ +\ 4\epsilon_0}\big)^M\\
&\leq\ 
\exp(-\tfrac{M\epsilon_0}{4\delta\ +\ 2 \epsilon_0}).   
\end{align*}
We can then set~$P_0 =
\exp(-\tfrac{M\epsilon_0}{4\delta\ +\ 2 \epsilon_0})$ so that finally
\begin{align}
&P_{\rm fail}\nonumber\\
&= \bb P\big[\cl Q_\delta[\bs \Phi \bs x + \bs \xi] = Q_\delta[\bs \Phi \bsxalt + \bs \xi]\,\big|\
\|\bs x - \bsxalt\| \geq \epsilon_0\big]\nonumber\\
&\leq \tinv{2}(\tfrac{3}{s})^N \exp(-\tfrac{M\epsilon_0}{4\delta\ +\ 2 \epsilon_0}) \nonumber\\
&= \tinv{2} \exp(N \log(\tfrac{24\sqrt N + 6}{\epsilon_0})
  -\tfrac{M\epsilon_0}{4\delta\ +\ 2 \epsilon_0}) \nonumber\\
\label{eq:pfail-grfcq}
&\leq \tinv{2} \exp(N \log(\tfrac{29 \sqrt{N}}{\epsilon_0}) -\tfrac{M\epsilon_0}{4\delta\ +\ 2 \epsilon_0}).
\end{align} 

Therefore, if we want~$P_{\rm fail} \leq \eta$ for some~$0<\eta<1$, it
suffices to impose 
$$
M  \geq \tfrac {4\delta\ +\
  2 \epsilon_0}{\epsilon_0}\,\big(N \log(\tfrac{29 \sqrt{N}}{\epsilon_0}) + \log \tfrac{1}{2\eta}\big), 
$$
which determines the condition invoked in Theorem~\ref{thm:consist-impose-closeness}.

Knowing that we have necessarily~$\epsilon_0\leq 2$ since
$\bs x,\bsxalt\in \bb B^{N}$, a stronger condition for
\eqref{eq:bounds-max-error-GRF} to occur with the same lower bound
on its probability reads 
\begin{equation}
  \label{eq:cond-M}
M  \geq \tfrac {4(\delta\ +\
  1)}{\epsilon_0}\,\big(N \log(\tfrac{29 \sqrt{N}}{\epsilon_0}) + \log \tfrac{1}{2\eta}\big). 
\end{equation}
Alternatively, saturating this condition, we have 
\begin{equation*}
\epsilon_0 = \tfrac {4(\delta\ +\
  1)}{M}\,\big(N \log(\tfrac{29\sqrt{N}}{\epsilon_0}) + \log
\tfrac{1}{2\eta}\big)\ 
\leq\  \tfrac {4(\delta\ +\
  1)}{M}\,\big(N \log(\tfrac{5M}{2\sqrt N}) + \log
\tfrac{1}{2\eta}\big).  
\end{equation*}
where we used the fact that, from~\eqref{eq:cond-M}, 
$$
\tfrac{M}{\sqrt N} \geq \tfrac{4}{\epsilon_0} \sqrt N \log
(\tfrac{29 \sqrt N}{\epsilon_0}) \geq \tfrac{4}{\epsilon_0}  
\log (\tfrac{29}{\sqrt 2}) \sqrt N
\geq \tfrac{2}{5} \tfrac{29}{\epsilon_0} \sqrt N,
$$ 
since~$\epsilon_0 \leq
2$ and assuming~$N \geq 2$. 

In other words, assuming~$\delta = O(1)$, there exists a constant
$C>0$ such that, 
$$
\bb P\big[ \cl E_\delta(\bs \Phi, \bs \xi, \bb B^N)\ \leq\ C \big( \tfrac {N}{M}\,\log(\tfrac{M}{\sqrt
  N}) + \tfrac {1}{M} \log
\tfrac{1}{2\eta}\big) \big] \geq 1 - \eta,
$$
which proves \eqref{eq:worst-case-error-GRFCQ} in Corollary~\ref{cor:bound-distance}.

\subsection{Quantized Compressed Sensing of Sparse Vectors}
\label{sec:extension-k-sparse}

We prove now Theorem~\ref{thm:consist-impose-closeness-saprse-case}
(and the QCS part of Corollary~\ref{cor:bound-distance}), \ie we adapt the minimal number of
measurements in the statement of
Theorem~\ref{thm:consist-impose-closeness} to the context of QCS when both the original
signal and the consistent reconstruction are additionally assumed to be~$K$-sparse in~$\bb B^N \subset \bb
R^N$, \ie they belong to~$\cl K = \Sigma_K\cap \bb B^N$ with~$\Sigma_K
:= \{\bs w \in \bb R^N: \|\bs w\|_0 \leq K\}$.

Notice first that, given a fixed support~$T_0 \subset [N]$ with~$\#T_0 =
2K$, thanks to the developments of Sec.~\ref{sec:prior-works}, 
\begin{multline*}
\bb  P\big[(\exists\bs x,\bsxalt \in \bb B^N: \|\bs x - \bsxalt\| \geq
\epsilon_0,\ \supp \bs x \cup \supp \bsxalt
\subset T_0):\
\cl Q_\delta[\bs \Phi \bs x + \bs \xi] = Q_\delta[\bs \Phi \bsxalt + \bs \xi]\big]\\
\leq \tinv{2} \exp(2K \log(\tfrac{29 \sqrt{2K}}{\epsilon_0})
-\tfrac{M\epsilon_0}{4\delta\ +\ 2 \epsilon_0}),  
\end{multline*}
since the subspace of vectors supported in~$T_0$ is equivalent to~$\bb
R^{2K}$.

Since there are no more than~${N \choose 2K} \leq
(\tfrac{eN}{2K})^{2K}$ choices of~$2K$-length supports in~$[N]$, another
union bound provides 
\begin{align}
&\bb P\big[(\exists\bs x,\bsxalt \in \bb B^N\cap\Sigma_K: \|\bs x -
  \bsxalt\| \geq \epsilon_0):  \cl Q_\delta[\bs \Phi \bs x + \bs \xi] = Q_\delta[\bs \Phi \bsxalt + \bs \xi]\,\big]\nonumber\\
&\leq \bb P\big[(\exists T\subset [N]: \#T = 2K),\ (\exists\bs x,\bsxalt \in \bb B^N: \|\bs x - \bsxalt\| \geq
  \epsilon_0,\ \supp \bs x \cup \supp \bsxalt
\subset T):\nonumber\\
&\pushright{\cl Q_\delta[\bs \Phi \bs x + \bs \xi] = \cl Q_\delta[\bs \Phi \bsxalt + \bs \xi]\,\big]}\nonumber\\
&\leq \tinv{2} {\ts {N \choose 2K}} \exp(2K \log(\tfrac{29 \sqrt{2K}}{\epsilon_0})
-\tfrac{M\epsilon_0}{4\delta\ +\ 2 \epsilon_0}) \nonumber\\
\label{eq:pfail-qcs}
&\leq\ \tinv{2} \exp(2K \log(\tfrac{29 e N}{\sqrt{2K}\epsilon_0})
-\tfrac{M\epsilon_0}{4\delta\ +\ 2 \epsilon_0}).    
\end{align}
Again, willing to have this last probability smaller than~$\eta \in (0,1)$
leads to imposing
$$
M \geq \tfrac {4\delta\ +\ 2 \epsilon_0}{\epsilon_0} \big(2K \log(\tfrac{29 e N}{\sqrt{2K}\epsilon_0}) + \log(\tfrac{1}{2\eta})\big),
$$
which, by noting that~$29e/\sqrt 2 < 56$, provides the key condition of Theorem
\ref{thm:consist-impose-closeness-saprse-case}.

Since~$\epsilon_0 \leq 2$, a stronger condition reads 
$$
M \geq \tfrac {4(\delta\ +\ 1)}{\epsilon_0} \big(2K \log(\tfrac{56 N}{\sqrt{K}\epsilon_0}) + \log(\tfrac{1}{2\eta})\big),
$$
which gives the crude estimation
\begin{equation*}
\tfrac{MN}{\sqrt{K^3}} \geq \tfrac {8N}{\epsilon_0\sqrt{K}}
\log(\tfrac{56 N}{\sqrt{K}\epsilon_0})
> \tfrac {8N}{\epsilon_0\sqrt{K}}
\log(\tfrac{56\sqrt N}{\epsilon_0}) > \tinv{2}\,\tfrac {56N}{\sqrt{K}\epsilon_0},
\end{equation*}
using~$K \leq N$ and~$N \geq 2$. Therefore, saturating the condition
on~$M$ above,
\begin{equation*}
\epsilon_0 = \tfrac {4(\delta\ +\ 1)}{M} \big(2K
\log(\tfrac{56 N}{\sqrt{K}\epsilon_0}) +
\log(\tfrac{1}{2\eta})\big)\ 
\leq  \tfrac {4(\delta\ +\ 1)}{M} \big(2K
\log(\tfrac{2MN}{\sqrt{K^3}}) +
\log(\tfrac{1}{2\eta})\big),  
\end{equation*}
which shows that, if~$\delta = O(1)$, there exists a constant~$C>0$
for which
\begin{equation*}
\bb P\big[ \cl E_\delta(\bs \Phi, \bs \xi, \Sigma_K \cap \bb B^N)\
\leq\ C \big( \tfrac {K}{M} 
\log(\tfrac{MN}{\sqrt{K^3}}) +
\tinv{M}\,\log(\tfrac{1}{2\eta})\big) \big] \geq 1 - \eta.  
\end{equation*}
This demonstrates \eqref{eq:worst-case-error-QCS} in Corollary~\ref{cor:bound-distance}.

\subsection{Proximity of Almost Consistent Signals}
\label{sec:prox-almost-cons}

As stated in the end of Sec.~\ref{sec:main-results}, the strict consistency between the quantized
projections of two vectors of~$\cl K \subset \bb R^N$ can be relaxed
while still keeping their maximal distance bounded. 
To show this, we follow a
similar procedure to that developed in~\cite{QIHT} for the case of 1-bit
quantized random projections. We may first observe that if 
\begin{equation}
  \label{eq:almost-consistency-cond}
  \|\cl Q_\delta(\bs \Phi \bs x + \bs
  \xi)-\cl Q_\delta(\bs \Phi \bsxalt + \bs \xi)\|_1\ \leq\ r\,\delta,  
\end{equation}
for some~$r \in \bb N$, at most~$r$ measurements differ between~$\cl Q_\delta(\bs \Phi \bs x + \bs
\xi)$ and~$\cl Q_\delta(\bs \Phi \bsxalt + \bs \xi)$. Thus, there exists a subset~$T$ of~$[M]$ with size at least~$M-r$ such that~$\cl R_{T} \cl Q_\delta(\bs \Phi \bs x + \bs
\xi) = \cl R_{T} \cl Q_\delta(\bs \Phi \bsxalt + \bs \xi)$, with
the corresponding restriction operator~$\cl R_{T}$ defined in
the Introduction.

Therefore, for~$\cl K \subset \bb R^N$ and denoting with~$[M]_r$ the set of
all subsets of~$[M]$ of size~$M-r$, a union bound provides
\begin{align*}
P_r:=&\ts \ \bb P\big[\exists\,\bs x, \bsxalt \in \cl K:  \|\bs x - \bsxalt\|\geq
\epsilon_0\\
&\qquad \text{s.t.}\ \| \cl Q_\delta(\bs \Phi \bs x + \bs
\xi)-\cl Q_\delta(\bs \Phi \bsxalt + \bs \xi)\|_1 \leq r\,\delta\,\big]\\
\leq&\ts\ \bb P\big[\exists T \subset [M]_r,\ \exists\,\bs x, \bsxalt \in \cl K:  \|\bs x - \bsxalt\|\geq
\epsilon_0\\
&\qquad \text{s.t.}\ \cl R_{T}\cl Q_\delta(\bs \Phi \bs x + \bs
\xi) = \cl R_{T}\cl Q_\delta(\bs \Phi \bsxalt + \bs \xi)\big]\\
\leq&\ts\ \sum_{T \subset [M]_r} \bb P\big[\exists\,\bs x, \bsxalt \in \cl K:  \|\bs x - \bsxalt\|\geq
\epsilon_0\\
&\qquad \text{s.t.}\ \cl Q_\delta(\cl R_{T_{\rm
    c}} \bs \Phi \bs x + \bs
\xi_{T}) = \cl Q_\delta(\cl R_{T} \bs \Phi \bsxalt + \bs
\xi_{T})\big].
\end{align*}

Each element of this last sum can be bounded from our
  developments of Sec.~\ref{sec:main-result} and of Sec.~\ref{sec:extension-k-sparse}.
In the GRFCQ case, \ie if~$\cl K = \bb B^N$ with~$M \geq N$, using~\eqref {eq:pfail-grfcq} and~${M \choose M-r} = {M
  \choose r}\leq (eM/r)^r$, we find
\begin{align*}
P_r&\leq\ \tinv{2}\,{\ts {M \choose M-r}} \exp(N \log(\tfrac{29 \sqrt{N}}{\epsilon_0}) -\tfrac{(M-r)\epsilon_0}{2\delta\ +\ 4 \epsilon_0})\\
&\leq\ \tinv{2} \exp(r \log(\tfrac{eM}{r})\ +\ N \log(\tfrac{29 \sqrt{N}}{\epsilon_0}) -\tfrac{(M-r)\epsilon_0}{4\delta\ +\ 2 \epsilon_0}).
\end{align*}
In the QCS case where~$\cl K = \Sigma_K \cap \bb B^N$, using
\eqref{eq:pfail-qcs} in Sec.~\ref{sec:extension-k-sparse}, we have similarly
\begin{align*}
P_r&\leq\ \tinv{2}\,{\ts {M \choose M-r}} \exp(2K \log(\tfrac{29 e N}{\sqrt{2K}\epsilon_0})
-\tfrac{(M-r)\epsilon_0}{4\delta + 2 \epsilon_0})\\
&\leq\ \tinv{2} \exp(r \log(\tfrac{eM}{r}) + 2K \log(\tfrac{29 e N}{\sqrt{2K}\epsilon_0})
 - \tfrac{(M-r)\epsilon_0}{4\delta + 2 \epsilon_0}).
\end{align*}
Imposing that those bounds on~$P_r$ be smaller than~$\eta \in (0,1)$,
we find that, as soon as  
\begin{equation}
  \label{eq:cond-M-GRFCQ}
  M \geq\ r + \tfrac {4\delta\ +\ 2 \epsilon_0}{\epsilon_0}\big(r
  \log(\tfrac{eM}{r}) + N \log(\tfrac{29 \sqrt{N}}{\epsilon_0}) + \log(\tfrac{1}{2\eta})\big),
\end{equation}
for GRFCQ, or
\begin{equation}
  \label{eq:cond-m-any-consist-support}
  M \geq\ r +
\tfrac {4\delta\ +\ 2\epsilon_0}{\epsilon_0} \big(r
  \log(\tfrac{eM}{r}) +
  2K \log(\tfrac{56 N}{\sqrt{K}\epsilon_0}) + \log(\tfrac{1}{2\eta})\big),  
\end{equation}
for QCS, and given~$\bs \Phi \sim \cl N^{M \times N}(0,1)$ and~$\bs \xi \sim
\cl U^{M}([0,\delta])$, the event
\begin{equation}
  \label{eq:almost-consist-inv-prox}
  \forall \bs x,\bsxalt\in \cl K,\quad \| \cl Q_\delta(\bs \Phi \bs x + \bs
  \xi)-\cl Q_\delta(\bs \Phi \bsxalt + \bs \xi)\|_1 \leq r\,\delta\quad
  \Rightarrow\quad\|\bs x - \bsxalt\|\leq \epsilon_0,  
\end{equation}
holds with probability higher than~$1 - \eta$, with $\cl K$ fixed as
above by the associated case. 

These considerations allow us to prove
  Theorems~\ref{thm:bound-distance-almost-consist} and~\ref{thm:propor-consist}, \ie to bound, respectively, the proximity of vectors
  whose quantized random projections are either ``\emph{almost perfectly
  consistent}'', \ie if~$r=O(1)$ relatively to~$M$, or for which the number of inconsistent projections is
  proportional to~$M$, what we call ``\emph{proportional inconsistency}''.

\subsubsection{Almost perfect consistency}
\label{sec:a-almost-perfect}

In this regime, we
  assume that~$r$ is bounded relatively to the possible increasing of
 ~$M$, \ie $r = O(1)$. In the context of GRFCQ and allowing a stronger condition on~$M$, we can then simplify~\eqref{eq:cond-M-GRFCQ} by a series of crude upper bounds and
  observe that
\begin{align*}
&r + \tfrac {4\delta + 2 \epsilon_0}{\epsilon_0}\big(r
\log(\tfrac{eM}{r}) + N \log(\tfrac{29 \sqrt{N}}{\epsilon_0}) + \log(\tfrac{1}{2\eta})\big)\\
&\leq \tfrac {4\delta + 2\epsilon_0}{\epsilon_0} \big( \tfrac{3}{2} r 
\log(eM) +  \tfrac{3}{2}N \log(\tfrac{29 \sqrt{N}}{\epsilon_0}) +
\log(\tfrac{1}{2\eta})\big)\\
&\leq \tfrac {4(\delta + 1)}{\epsilon_0} \big( \tfrac{3}{2} r 
\log(\tfrac{2eM}{\epsilon_0}) +  \tfrac{3}{2}N \log(\tfrac{29 \sqrt{N}}{\epsilon_0}) +
\log(\tfrac{1}{2\eta})\big)\\
&\leq \tfrac {4(\delta + 1)}{\epsilon_0} \big( \tfrac{3}{2} ( N + r) 
\log(\tfrac{29M}{\epsilon_0}) +
\log(\tfrac{1}{2\eta})\big)\\
&\leq \tfrac {6(\delta + 1)}{\epsilon'_0} N \log(\tfrac{29M}{\epsilon'_0}) +
\tfrac {4(\delta + 1)}{\epsilon'_0}\log(\tfrac{1}{2\eta}),
\end{align*}
using the variable change~$\epsilon_0 = \frac{N+r}{N} \epsilon_0' \geq
\epsilon_0'$,
$\epsilon_0 \leq 2$ and~$M \geq N \geq \sqrt N$. Notice that in the case where~$r=0$, remembering that
the term~$r\log(eM/r)$ above comes from a bound on~$\log {M \choose r}$, we can
assume~$r\log(eM/r) = 0$, and since~$r\in\Nbb$, we can write
$r\log(eM/r) \leq r \log(eM)$.

For the case of QCS, starting from the RHS of
\eqref{eq:cond-m-any-consist-support}, we get
\begin{align*}
&r + \tfrac {4\delta + 2\epsilon_0}{\epsilon_0} \big(r
\log(\tfrac{eM}{r}) +  2K \log(\tfrac{56 N}{\sqrt{K}\epsilon_0}) + \log(\tfrac{1}{2\eta})\big)\\
&\leq \tfrac {4\delta + 2\epsilon_0}{\epsilon_0} \big( 2 r 
\log(eM) +  2K \log(\tfrac{56 N}{\sqrt{K}\epsilon_0}) +
\log(\tfrac{1}{2\eta})\big)\\
&\leq \tfrac {4(\delta + 1)}{\epsilon_0} \big( 2 r 
\log(\tfrac{2eM}{\epsilon_0}) +  2K \log(\tfrac{56 N}{\sqrt{K}\epsilon_0}) +
\log(\tfrac{1}{2\eta})\big)\\
&= \tfrac {4K(\delta + 1)}{(K+r)\epsilon'_0} \big( 2 r 
\log(\tfrac{2eMK}{(K+r)\epsilon'_0}) +  2K \log(\tfrac{56 N \sqrt{K}}{(K+r)\epsilon'_0}) +
\log(\tfrac{1}{2\eta})\big)\\
&\leq \tfrac {4K(\delta + 1)}{(K+r)\epsilon'_0} \big( 2 r 
\log(\tfrac{2eM}{\epsilon'_0}) +  2K \log(\tfrac{56 N}{\epsilon'_0}) +
\log(\tfrac{1}{2\eta})\big)\\
&\leq \tfrac {4K(\delta + 1)}{(K+r)\epsilon'_0} 2 (r + K)
  \log(\tfrac{56 \max(N,M/10)}{\epsilon'_0}) + \tfrac {4(\delta + 1)}{\epsilon'_0} \log(\tfrac{1}{2\eta})\\
&= \tfrac {8(\delta + 1)}{\epsilon'_0}  K \log(\tfrac{56 \max(N,M/10)}{\epsilon'_0}) +
\tfrac {4(\delta + 1)}{\epsilon'_0} \log(\tfrac{1}{2\eta}),
\end{align*}
using now the variable change~$\epsilon_0 = \frac{K+r}{K} \epsilon_0'$,
$\epsilon_0 \leq 2$ and
$2e/56 < 1/10$, and with the same remark on the vanishing value of~$r
\log(eM/r)$ when~$r=0$.

Therefore, rewriting everything as a function of~$\epsilon'_0$ in~\eqref{eq:almost-consist-inv-prox} and
forgetting the prime symbol, we find that, as soon as  
\begin{equation}
  \label{eq:simp-cond-M-GRFCQ}
M \geq \tfrac {6(\delta\ +\ 1)}{\epsilon_0} N \log(\tfrac{29M}{\epsilon_0})\ +\
\tfrac {4(\delta\ +\ 1)}{\epsilon_0}\log(\tfrac{1}{2\eta}),
\end{equation}
for GRFCQ, or if 
\begin{equation}
  \label{eq:simp-cond-M-QCS}
  M \geq \tfrac {8(\delta\ +\ 1)}{\epsilon_0}  K \log(\tfrac{56 \max(N,M/10)}{\epsilon_0}) +
  \tfrac {4(\delta\ +\ 1)}{\epsilon_0} \log(\tfrac{1}{2\eta}),  
\end{equation}
for QCS, and given a draw of~$\bs \Phi \sim \cl N^{M \times N}(0,1)$ and~$\bs \xi \sim
\cl U^{M}([0,\delta])$, the event
\begin{equation}
  \label{eq:tmp1}
\forall \bs x,\bsxalt\in \cl K,\quad \| \cl Q_\delta(\bs \Phi \bs x + \bs
\xi)-\cl Q_\delta(\bs \Phi \bsxalt + \bs \xi)\|_1 \leq r\,\delta\quad
\Rightarrow\quad\|\bs x - \bsxalt\|\leq c_r\,\epsilon_0,  
\end{equation}
holds with probability higher than~$1 - \eta$, with~$\cl K = \bb B^N$
and~$c_r=\tfrac{N+r}{N}$ for GRFCQ and with~$\cl K = \Sigma_K \cap \bb
B^N$ and~$c_r = \tfrac{K+r}{K}$ for QCS.

In the case of GRFCQ, saturating the condition on~$M$ above we find
$$
\epsilon_0 = \tfrac {6(\delta\ +\ 1)}{M} N \log(\tfrac{29 M}{\epsilon_0})\ +\
\tfrac {4(\delta\ +\ 1)}{M}\log(\tfrac{1}{2\eta}).
$$
Using~$\epsilon_0 \leq 2$ from~$\cl K \subset \bb B^N$, this
saturation involves~$\epsilon_0 \geq 10/M$ so that 
$$
\epsilon_0 \leq \tfrac {6(\delta\ +\ 1)}{M} N \log(3 M^2)\ +\
\tfrac {4(\delta\ +\ 1)}{M}\log(\tfrac{1}{2\eta}).
$$

Therefore, from~\eqref{eq:tmp1} with~$c_r = (N+r)/N$, if~$\delta = O(1)$, there is a~$C>0$ such that
\begin{equation*}
\bb P\big[\cl E^r_{\delta}(\bs \Phi, \bs \xi, \bb B^N) \leq C \tfrac{N+r}{M} \big(  \log(M) +
\tinv{N}\log(\tfrac{1}{2\eta}) \big) \big]\ \geq 1 - \eta,  
\end{equation*}
which proves~\eqref{eq:worst-case-error-GRFCQ-almost-consist} in Theorem~\ref{thm:bound-distance-almost-consist}.

Finally, in the case of QCS, 
since for the~$M$ saturating~\eqref{eq:simp-cond-M-QCS} we have 
\begin{equation*}
M \geq \tfrac {8(\delta\ +\ 1)}{\epsilon_0} K \log(\tfrac{56
  \max(N,M/10)}{\epsilon_0})\
 \geq \tfrac {8}{\epsilon_0} K \log(56)
> \tinv{2}\,\tfrac{56}{\epsilon_0} K,  
\end{equation*}
we find 
\begin{align*}
\epsilon_0&= \tfrac {8(\delta + 1)}{M}  K \log(\tfrac{56 \max(N,M/10)}{\epsilon_0}) +
\tfrac {4(\delta + 1)}{M} \log(\tfrac{1}{2\eta})\\
&< \tfrac {8(\delta + 1)}{M}  K \log(\tfrac{2M \max(N,M/10)}{K}) +
\tfrac {4(\delta + 1)}{M} \log(\tfrac{1}{2\eta}).  
\end{align*}
Consequently, using~\eqref{eq:tmp1} with~$c_r = (K+r)/K$, if~$\delta = O(1)$, there exists a~$C>0$ such that
\begin{equation*}
\bb P\big[\cl E^r_{\delta}(\bs \Phi, \bs \xi, \Sigma_K \cap \bb B^N)
\
\leq C \tfrac {K+r}{M}  \big (\log(\tfrac{M \max(N,M)}{K}) +
\log(\tfrac{1}{2\eta}))\big]\ \geq 1 -\eta,  
\end{equation*}
which justifies~\eqref{eq:worst-case-error-sparse-case-almost-consist}
in Theorem~\ref{thm:bound-distance-almost-consist}.

\subsubsection{Proportional inconsistency}
\label{sec:prop-inconsist}

We now prove Theorem~\ref{thm:propor-consist} and consider that~$r$ is actually proportional to~$M$, \ie
there exists a constant~$0<\rho<1$ such that~$r\leq \rho M$
in~\eqref{eq:almost-consistency-cond}. In words, this could happen if the
number of inconsistent quantized projections between those of~$\bs x$
and~$\bsxalt$ represents a constant proportion of~$M$, \ie $r/M =
O(1)$. Coming back to~\eqref{eq:cond-M-GRFCQ} and
\eqref{eq:cond-m-any-consist-support} and assuming $r=\rho M$, we
easily get the equivalent conditions
\begin{equation*}
\ts M
\geq \frac{4\delta +
  2\epsilon_0}{(1-\rho(1+2\log(e/\rho)))\,\epsilon_0- 4\rho \delta
  \log(e/\rho)}\
\ts\big(N \log(\frac{29 \sqrt
  N}{\epsilon_0}) + \log(\frac{1}{2\eta})\big),  
\end{equation*}
 for GRFCQ, and 
 \begin{equation*}
\ts M \geq \tfrac {4\delta\ +\ 2\epsilon_0}{\epsilon_0(1 -
  \rho(1+2\log(e/\rho))) - 4\rho\delta \log(e/\rho)}\
\ts
\big( 2K \log(\tfrac{56 N}{\sqrt{K}\epsilon_0}) + \log(\tfrac{1}{2\eta})\big).     
 \end{equation*}
for QCS.

Therefore, assuming 
$$
\bar\rho := \rho\,(1+2\log(e/\rho)) < 1,
$$ 
which is satisfied if~$\rho < 1/10$, and defining~$\epsilon_0' =
(1-\bar\rho)\,\epsilon_0- 4\rho
\delta \log(e/\rho) \leq \epsilon_0$, \ie $\epsilon_0 = (1-\bar\rho)^{-1} (\epsilon_0' + 4\rho
\delta \log(e/\rho) )$, we find that if 
\begin{equation}
  \label{eq:cond-M-inconsistent-GRFCQ}
  \ts M \geq \tfrac{4\delta + 4}{\epsilon_0'}\big(N \log(\frac{29 \sqrt  N}{\epsilon'_0}) + \log(\frac{1}{2\eta})\big),
\end{equation}
for GRFCQ, or 
\begin{equation}
  \label{eq:cond-M-inconsistent}
  \ts M \geq \tfrac {4\delta\ +\ 2}{\epsilon'_0} \big( 2K \log(\tfrac{56 N}{\sqrt{K}\epsilon'_0}) + \log(\tfrac{1}{2\eta})\big),
\end{equation}
for QCS, we have, with probability at least~$1-\eta$, the event
\begin{equation}
  \label{eq:tmp1-prop-inconst}
\forall \bs x,\bsxalt\in \cl K, \| \cl Q_\delta(\bs \Phi \bs x + \bs
\xi)-\cl Q_\delta(\bs \Phi \bsxalt + \bs \xi)\|_1 \leq \rho\,\delta\,M\quad
\Rightarrow\quad\|\bs x - \bsxalt\|\leq C_\rho \epsilon'_0 + D_\rho \delta,
\end{equation}
with~$C_\rho := (1-\bar\rho)^{-1} \geq 1$ and~$D_\rho :=
(1-\bar\rho)^{-1} 4 \rho \log(e/\rho)$, and~$\cl K$ set
to~$\bb B^N$ for GRFCQ and to~$\Sigma_K \cap \bb
B^N$ for QCS.  

This last relation shows that there is a price to pay when the
inconsistency between the quantized projections of~$\bs x$ and
$\bsxalt$ reaches a level that is proportional to~$M$. While the first term~$C_\rho \epsilon_0$ can be made arbitrarily
low by increasing~$M$, the second term
$D_\rho\delta \geq 4\rho\delta$ is constant and fixed by~$\rho$ and~$\delta$. This part
vanishes only when~$\rho$ tends to 0 (see also Remark~\ref{rem:tight-bias}), while~$C_\rho$ approaches~1 in
this~case.

\section*{Acknowledgements}
\label{sec:acknowledgements}

The author thanks Alexander Powell (Vanderbilt U., USA)
and Jalal Fadili (GREYC, U. Caen, France) for the inspiring discussions
made during the ICCHA5 conference at Vanderbilt University (TN, USA),
and Valerio Cambareri (UCLouvain, Belgium) for his advice on the
writing of this paper. The author thanks also the anonymous reviewers
for their useful advices and remarks for improving the structure,
the presentation and the discussion of this paper. Laurent Jacques is a Research Associate funded by the Belgian
F.R.S.-FNRS. 

\appendix

\section{Proof of Lemma~\ref{lemma:main-proba-bound}}
\label{sec:proof-lemma}

This appendix is dedicated to the proof of
Lemma~\ref{lemma:main-proba-bound}. This one lies at the heart of all our
developments as it determines both Theorems~\ref{thm:consist-impose-closeness} and
\ref{thm:consist-impose-closeness-saprse-case} and their corollaries.
In short, given the dithered quantized mapping
\eqref{eq:qcs-unif-dith-model} and two non-overlapping balls centered
on two distinct vectors, this lemma
bounds the probability that there exist two consistent
vectors, one in each ball, and relates this bound the distance between
the ball centers and the ball width.  The reason why this lemma is
important is due to the fact it allows us a certain form of
continuity in the proximity analysis of consistent vectors. This point
is mandatory for
proving Theorems~\ref{thm:consist-impose-closeness} and
\ref{thm:consist-impose-closeness-saprse-case} by covering the
signal domain with balls of appropriate radius, hence allowing us to
use a union bound argument for studying the proximity of any
consistent vector pairs in this space.

Let us recall the context of this lemma. We want to show that, given two
points~$\tilde{\bs p}, \tilde{\bs
  q} \in \bb R^N$, there exists a radius~$s' \geq \tfrac{1}{8\sqrt N} \|\tilde{\bs p}-\tilde{\bs
  q}\|$ such that, for~$\bs \Phi \sim \cl N^{N \times
  M}(0,1)$ and~$\bs \xi \sim \cl U^{M\times 1}([0,\delta])$, the probability  
\begin{equation*}
P_{s'}(\alpha, M)\ :=
 \bb P\big[\exists \bs u \in \bb B_{s'}(\tilde{\bs p}), \exists \bs v \in
\bb B_{s'}(\tilde{\bs q}),\
\cl Q_\delta[\bs \Phi \bs u + \bs \xi] =\cl
Q_\delta[\bs \Phi \bs v + \bs \xi]\big]  
\end{equation*}
satisfies
$$  
P_{s'}(\alpha, M)\ \leq\ \big(1 - \tfrac{3\alpha}{8\ +\ 4\alpha}\big)^M,  
$$
with~$\alpha = {\|\tilde{\bs p}-\tilde{\bs q}\|}/{\delta}$.

Notice first that we can focus on upper bounding the probability associated to a single projection by the random vector
$\bs\varphi \sim \cl N^{N\times 1}(0,1)$ quantized with~$\cl Q_\delta$
with a scalar
dithering~$\xi \sim \cl U([0,\delta])$. The result for
$M$ dithered quantized projections will simply
follow by raising the single measurement bound to the power~$M$,
\ie $P_{s'}(\alpha, M) \leq (P_{s'}(\alpha, 1))^M$. 

We write~$\bs \varphi = \phi\,\hat{\bs \varphi}$, where
$\hat{\bs \varphi} \in \bb S^{N-1}$ is
uniformly distributed at random over~$\bb S^{N-1}$ and the length
$\phi=\|\bs\varphi\|\sim \chi(N)$
follows a~$\chi$ distribution with~$N$ degrees of freedom. We are
going first to estimate the following conditional probability:
\begin{align}
P_{s'}(\alpha, 1|\phi)\nonumber&:=\bb P\big[\exists \bs u \in \bb B_{s'}(\tilde{\bs p}), \exists \bs v \in
\bb B_{s'}(\tilde{\bs q}), \cl Q_\delta[\bs \varphi^T \bs u + \xi] =\cl Q_\delta[\bs \varphi^T \bs v + \xi]\,\big|\ \|\bs\varphi\|=\phi\big]\nonumber\\
&=\ \bb P\big[\exists \bs u \in \bb B_{s'}(\tilde{\bs p}), \exists \bs v \in
\bb B_{s'}(\tilde{\bs q}), \cl Q_\delta[\phi \hat{\bs \varphi}^T \bs u + \xi] =\cl Q_\delta[\phi\hat{\bs
\varphi}^T \bs v + \xi]\,\big|\ \|\bs \varphi\| = \phi\big]\nonumber\\
&=\ \bb P\big[\exists \bs u \in \bb B_{r}(\bs p), \exists \bs v \in
\bb B_{r}(\bs q), \cl Q_\delta[\hat{\bs \varphi}^T \bs u + \xi] =\cl Q_\delta[\hat{\bs
\varphi}^T \bs v + \xi]\,\big|\ \|\bs \varphi\| = \phi\big],\label{eq:tmpeq}
\end{align}
with the variable changes~$r=\phi s'$,~$\bs p = \phi \tilde{\bs p}$,~$\bs q = \phi \tilde{\bs
  q}$,~$\phi\hat{\bs \varphi}=\bs \varphi$. 
Notice that~$2r/\|\bs p - \bs q\| = 2s'/\|\tilde{\bs p} - \tilde{\bs
  q}\|$. Let us focus on this last probability, keeping in mind the
relationships between these parameters for estimating later a result
which is not conditioned to the knowledge of~$\phi$.

We follow the procedure described in
\cite{jacques2013quantized}. In this work, from a generalization of
the Buffon's needle problem~\cite{buffon_origin,hey2010georges} in~$N$
dimensions, it is shown that when
$r=0$, \ie when~$\bs u = \bs p$ and~$\bs v = \bs q$, computing
$P_{s'}(\alpha, 1|\phi)$ above is equivalent to estimating the probability that a segment (or
\emph{needle}) of
length~$L=\|\bs p - \bs q\|$ uniformly ``thrown'' at
random in~$\bb
R^N$, both spatially and in orientation, does not intersect a
fixed set 
of parallel~$(N-1)$-dimensional hyperplanes spaced by a distance
$\delta$.  

More precisely, given~$\hat{\bs \varphi} \in \bb S^{N-1}$ and~$\xi \in
[0, \delta]$, the function~$f(\bs
v) := \cl
Q_\delta(\hat{\bs \varphi}^T\bs v + \xi)$ is piecewise
constant in~$\bb R^N$ and the frontiers where its value changes correspond
to a set of parallel~$(N-1)$-dimensional hyperplanes in~$\bb R^N$. 
These hyperplanes are equi-spaced with a separating distance
$\delta$ and they are all normal to the direction~$\hat{\bs
  \varphi}$. Consequently, the quantity~$X := \inv{\delta}\big(\cl
Q_\delta(\hat{\bs \varphi}^T\bs p + \xi) - \cl
Q_\delta(\hat{\bs \varphi}^T\bs q + \xi)\big) \in \bb Z$ counts
the number of such hyperplanes intersecting the segment~$\overline{\bs
  p\bs q}$.  In this scenario, this segment is thus fixed and the
hyperplanes are randomly oriented and shifted by~$\hat{\bs \varphi}$
and~$\xi$, respectively. 

However, we can reverse the point of view and
rather consider those hyperplanes as fixed and normal, \eg to the first 
canonical axis~$\bs e_1$ of~$\bb R^N$. This is allowed by considering
the affine mapping~$\cl A_{\hat{\bs \varphi},\xi}: \bb R^N \to \bb
R^N$ implicitly defined by any combination
  of a rotation and of a translation in~$\bb R^N$ such that~$\bs e_1^T\cl A_{\hat{\bs \varphi},\xi}(\bs v) = \hat{\bs \varphi}^T \bs v + \xi$ for all
 ~$\bs v \in \bb R^N$. In words, thanks to~$\cl A_{\hat{\bs
      \varphi},\xi}$, projecting a point~$\bs v \in \bb R^N$ onto the random orientation~$\hat{\bs \varphi}$
  and shifting the result by~$\xi$ is equivalent to projecting the random point~$\cl
  A_{\hat{\bs \varphi},\xi}(\bs v)$ onto~$\bs e_1$. 

Therefore, denoting~$\bs p' = A_{\hat{\bs \varphi},\xi}(\bs p)$ and
$\bs q' = A_{\hat{\bs \varphi},\xi}(\bs q)$, it is easy to see that
the~$L$-length segment~$\overline{\bs
  p'\bs q'}$, \ie our \emph{needle}, is then oriented uniformly at random over~$\bb S^{N-1}$
while the distance of
its centrum~$\tinv{2}(\bs p' + \bs q')$ to the closest hyperplane follows a uniform random variable
over the interval~$[0, \delta/2]$. Moreover, we have
$$
X = \tinv{\delta}\big(\cl
Q_\delta(\bs e_1^T\bs p') - \cl
Q_\delta(\bs e_1^T\bs q')\big),
$$
so that~$X$ actually measures the number of intersections the segment~$\overline{\bs
  p'\bs q'}$ makes with the set of hyperplanes~$\cl G_\delta =
\bigcup_{k \in \bb Z} \{\bs x: \bs e_1^T \bs x = k\}$. In
\cite{jacques2013quantized}, the distribution of the discrete bounded random
variable~$X$ is actually fully
determined and denoted~${\rm Buffon}(L/\delta, N)$.

\begin{figure}
  \centering
  \includegraphics[width=5cm]{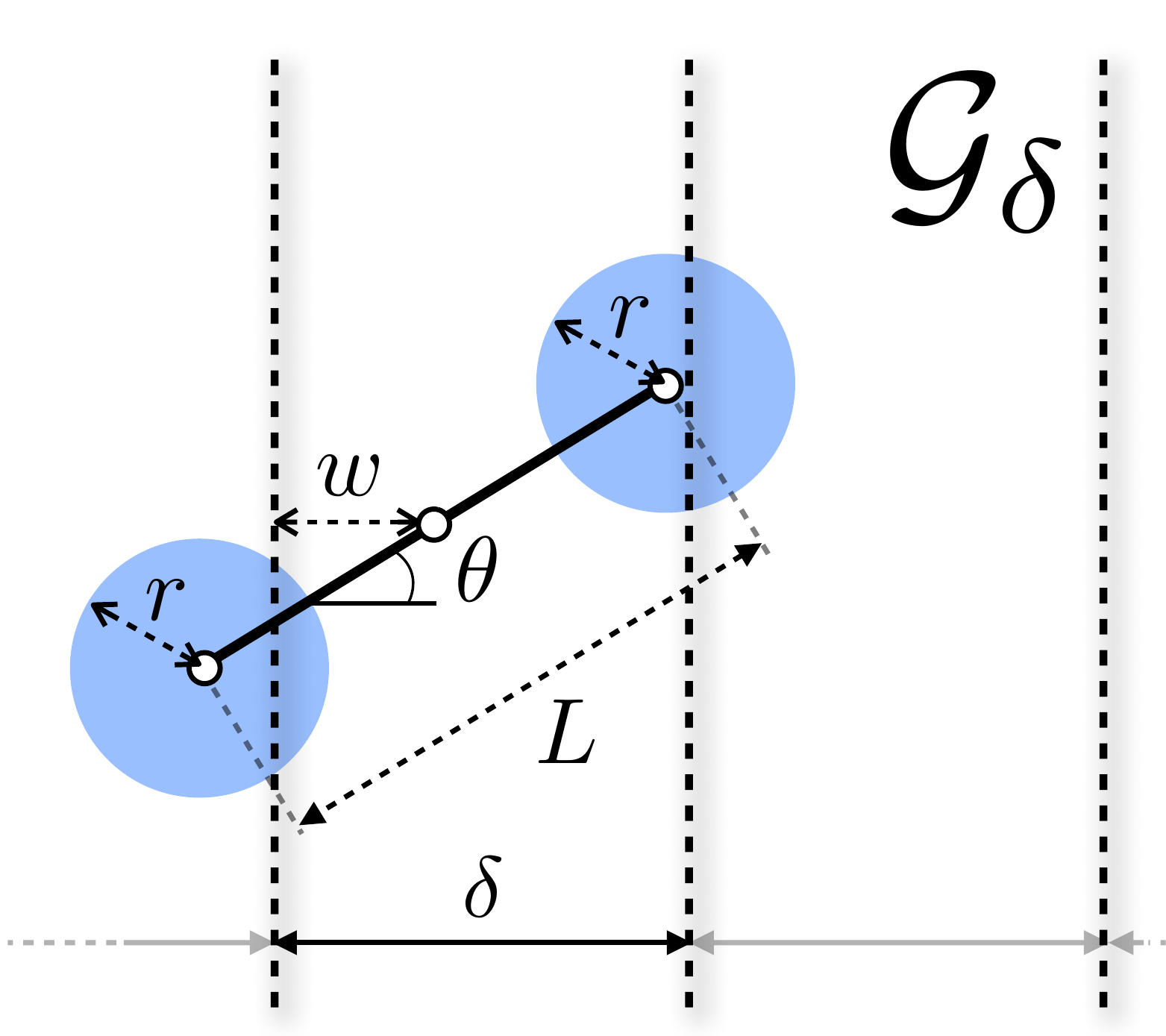}
  \caption{A Buffon ``dumbbell'' problem in 2-D.}
  \label{fig:dumbbell-problem}
\end{figure}

For~$r > 0$, Eq.\eqref{eq:tmpeq} shows that we must now consider the
two neighboring~$\ell_2$-balls of~$\bs p$ and~$\bs q$
in~$P_{s'}(\alpha, 1|\phi)$ and estimate the probability that at least
two points of these balls share the same dithered quantized projection onto
$\hat{\bs \varphi}$. Following the same argument as above, this new
problem is now equivalent to a new
Buffon experiment if the
previous needle is ended with two balls. In other words, we create a
\emph{dumbbell} shape formed by a segment of length~$L$
on the extremities of which two balls of radius~$r$ are centered  (see
Fig.~\ref{fig:dumbbell-problem}). 

It is then easy to see that~$P_{s'}(\alpha, 1|\phi)$ is equivalent
to the probability that there is no hyperplane of~$\cl G_\delta$ intersecting only the
part of the segment outside of the two balls when the dumbbell is thrown
randomly in~$\bb R^N$ as for previous Buffon's needle. Otherwise, having such an
intersection would mean that no pair of points
(taken in distinct balls) lie in the same subvolume delimited by
two consecutive hyperplanes, \ie they do not have the same quantized
projection, and conversely.

Let us parametrize this dumbbell by its distance~$w \sim \cl U([0, \delta/2])$
(estimated from the middle of the segment) to the closest hyperplane~$\cl
G_\delta$ and by its orientation drawn uniformly at random in~$\bb
S^{N-1}$. By symmetry, only the
angle~$\theta \in [0, \pi]$ made by the dumbbell with the normal
vector~$\bs e_1$ to~$\cl G_\delta$ is important
in this parametrization~\cite{jacques2013quantized}. Moreover, from
Fig.~\ref{fig:dumbbell-problem},  the
absence of intersection amounts to imposing 
$w \geq \inv{2}\,L|\cos\theta| - r$. The probability~$P_{s'}(\alpha,
1|\phi)$ is thus obtained by
\begin{align*}
\ts P_{s'}(\alpha, 1|\phi)&\ts= \int_0^{\pi} \kappa_N\,(\sin\theta)^{N-2}\,\ud\theta \int_0^{\delta/2} \bb I(w \geq \tfrac{L}{2}|\cos\theta| - r)
\tfrac{2}{\delta}\,\ud w,\\
&\ts= \tfrac{4\kappa_N}{\delta} \int_0^{\pi/2} (\sin\theta)^{N-2}
\ud\theta \int_0^{\delta/2} \bb I(w \geq \tfrac{L}{2}\cos\theta - r)
\,\ud w,
\end{align*}
where 
$\kappa_N (\sin\theta)^{N-2} \ud\theta$ is the area (normalized to
the one of~$\bb S^{N-1}$) of the thin spherical segment~$\cl S_{\ud \theta}(\theta):=\{\hat{\bs
  v} \in \bb S^{N-1}: \arccos(\bs e_1^T \hat{\bs
  v}) \in [\theta, \theta + \ud\theta]\}$, where~$\kappa_N :=
\frac{\Gamma(\frac{N}{2})}{\sqrt\pi\,\Gamma(\frac{N-1}{2})} =
B(\inv{2},\frac{N-1}{2})^{-1}$ and
$B(k,l)=\Gamma(k)\Gamma(l)/\Gamma(k+l)$ is the Beta function.

It is important to remark that, from~\cite{jacques2013quantized,qi2012bounds},
\begin{equation}
  \label{eq:bound-tau_N}
  \ts\ \tfrac{\sqrt 2}{\sqrt \pi} \, (N+1)^{-\frac{1}{2}}
  \ \leq\ \tfrac{2\kappa_N}{N-1}\ \leq\ \tfrac{\sqrt 2}{\sqrt \pi} \,  (N-1)^{-\frac{1}{2}},
\end{equation}
so that, for~$N \geq 2$,
\begin{equation*}
\tfrac{1}{\sqrt{2\pi}} \, (N+1)^{\frac{1}{2}} - 1\ <\ \kappa_N\ \leq\ \tfrac{1}{\sqrt{2\pi}} \,
  (N-1)^{\frac{1}{2}}\quad
\Rightarrow\ \kappa_N = \Theta(\sqrt{\tfrac{N}{2\pi}}).  
\end{equation*}

Let us define two angles~$0 \leq \theta_0 \leq \theta_1 \leq
\pi/2$ such that~$\cos\theta_0 =
\min(\tfrac{\delta + 2r}{L},1)$ and~$\cos\theta_1 = \tfrac{2r}{L}$,
assuming~$2r \leq L$ (otherwise,~$P_{s'} = 1$). The angular
integration domain can be split in three
intervals:~$[0,\theta_0]$,~$[\theta_0,\theta_1]$ and
$[\theta_1,\pi/2]$.  Over the first interval, the integral is always zero since, either we have a zero
measure interval ($\theta_0=0$) or~$\bb I(w \geq
\tfrac{L}{2}\cos\theta - r)=0$ since~$\tfrac{L}{2}\cos\theta\geq
\tfrac{L}{2}\cos\theta_0 = \delta/2 + r$ and~$0 \leq w \leq
\delta/2$. Moreover, over the last interval~$[\theta_1,\pi/2]$,~$I(w \geq
\tfrac{L}{2}\cos\theta - r) = 1$

Therefore, writing~$a=L/\delta$,
\begin{align*}
\ts P_{s'}(\alpha, 1|\phi)&\ts=\ \tfrac{4\kappa_N}{\delta} \int_{\theta_0}^{\theta_1} (\sin\theta)^{N-2}
(\tfrac{\delta}{2} -
\tfrac{L}{2}\cos\theta + r)\,\ud\theta\ +\
2\kappa_N \int_{\theta_1}^{\pi/2} (\sin\theta)^{N-2}\,\ud\theta\\ 
&\ts=1 + \tfrac{4\kappa_N}{\delta} \int_{\theta_0}^{\theta_1} (\sin\theta)^{N-2}
(\tfrac{\delta}{2} -
\tfrac{L}{2}\cos\theta + r)\,\ud\theta\ -\ 
2\kappa_N \int_{0}^{\theta_1} (\sin\theta)^{N-2}\,\ud\theta\\
&\ts=1 - \tfrac{4\kappa_N}{\delta} \int_{\theta_0}^{\theta_1} (\sin\theta)^{N-2}
(\tfrac{L}{2}\cos\theta - r)\,\ud\theta\ - 
\tfrac{4\kappa_N}{\delta} \int_{0}^{\theta_0} (\sin\theta)^{N-2}\,\tfrac{\delta}{2}
\,\ud\theta\\
&\ts=1\ -\ \tfrac{4\kappa_N}{\delta} \int_{0}^{\theta_1} (\sin\theta)^{N-2}
(\tfrac{L}{2}\cos\theta - r)\,\ud\theta\ + \ \tfrac{4\kappa_N}{\delta} \int_{0}^{\theta_0} (\sin\theta)^{N-2}
(\tfrac{L}{2}\cos\theta - (r + \tfrac{\delta}{2}))\,\ud\theta\\
&\ts=1\ -\ 2\kappa_N a\,\int_0^1
(1-v^2)^{\frac{N-3}{2}} \big[ (v - \tfrac{2r}{L})_+ - (v - \tfrac{2r+\delta}{L})_+\big]\ud v,
\end{align*}
applying a variable change~$v=\cos\theta$ on the last line.

Let us study this last integral and the function~$f(v) = (v -
\tfrac{2r}{L})_+ - (v - \tfrac{2r+\delta}{L})_+$. We can verify that
$F(v) :=  \int_0^{v} f(v') \ud v'$ is convex and reads
\begin{equation*}
2 F(v)\ =\ (v - \tfrac{2r}{L})^2_+ - (v -
\tfrac{2r+\delta}{L})^2_+\ =\
\begin{cases}0,&\text{if } v\leq \tfrac{2r}{L},\\
(v - \tfrac{2r}{L})^2,&\text{if } \tfrac{2r}{L} < v\leq \tfrac{2r+\delta}{L},\\
\tfrac{\delta}{L}(2v - \tfrac{4r+\delta}{L}),&\text{if}\ v > \tfrac{2r+\delta}{L}.
\end{cases}  
\end{equation*}
Moreover, by integrating by part,
\begin{equation*}
\ts \int_0^1
(1-v^2)^{\frac{N-3}{2}} f(v)\,\ud v\
\ts = \int_0^1
(N-3)\,v\,(1-v^2)^{\frac{N-5}{2}} F(v)\,\ud v,  
\end{equation*}
The positive measure~$\mu(v) =
(N-3)\,v\,(1-v^2)^{\frac{N-5}{2}}$ has unit mass over~$[0,1]$ so that,
by convexity of~$F$ and using Jensen's inequality, 
$$
\int_0^1
F(v)\,\mu(v)\,\ud v\ \geq\ F\big(\ts\int_0^1
v\,\mu(v)\,\ud v\big).
$$
However, since
\begin{equation}
  \label{eq:other-def-of-dN}
  \ts (N-3) \int_{0}^{1} (1 -
v^2)^{\frac{N-5}{2}}\,v^{q}\,\ud v =\
 \ts
\tfrac{N-3}{2}\,B(\tfrac{q+1}{2},\tfrac{N-3}{2}) =
\tfrac{\Gamma\big(\tfrac{q+1}{2}\big)\,\Gamma(\frac{N-1}{2})}{\Gamma\big(\tfrac{N+q-2}{2}\big)},
\end{equation}
we find
\begin{equation}
  \label{eq:other-def-of-dN}
  \ts \ts\int_0^1
v\,\mu(v)\,\ud v  = (N-3) \int_{0}^{1} (1 -
v^2)^{\frac{N-5}{2}}\,v^{2}\,\ud v\ =
\tfrac{\sqrt{\pi}\,\Gamma(\frac{N-1}{2})}{2\,\Gamma(\frac{N}{2})} =
\tfrac{1}{2\kappa_N},
\end{equation}
and 
$$
P_{s'}(\alpha, 1|\phi)\ \leq\ 1 - 2\kappa_N a\,F(\tinv{2\kappa_N}).
$$

From the definition of~$F$ above, if~$\tinv{2\kappa_N} \leq \frac{2r}{L}$,
$F=0$ and we cannot show anything. Let us thus set~$2r = \tfrac{\lambda}{2\kappa_N} L$, where
$\lambda\in (0,1)$ will be determined later. Notice that, since
$s'=\phi\,r$ and~$\|\tilde{\bs p}-\tilde{\bs q}\|=\phi\,L$, we
implicitly impose
$2r/L = 2s'/\|\tilde{\bs p}-\tilde{\bs
  q}\|=\tfrac{\lambda}{2\kappa_N}$. 

Then, 
$$
2F(\tinv{2\kappa_N}) = \begin{cases}\tinv{4\kappa^2_N}(1 - \lambda)^2,&\text{if } 
a \leq \tfrac{2\kappa_N}{1-\lambda},\\
\tinv{a}(\tfrac{1}{\kappa_N}(1-\lambda) - \tinv{a}),&\text{if } a \geq \tfrac{2\kappa_N}{1-\lambda},
\end{cases}
$$
so that, writing~$\phi_0 = \tfrac{2\kappa_N}{1-\lambda}$, we have
$$
P_{s'}(\alpha, 1|\phi)\ \leq\ \begin{cases}1 - \tfrac{a}{2 \phi_0} (1-\lambda),&\text{if } 
a \leq \phi_0,\\
\lambda + \tfrac{\phi_0}{2a}(1-\lambda),&\text{if } a \geq \phi_0.
\end{cases}
$$

Let us recall that~$P_{s'}(\alpha, 1|\phi)$ is defined conditionally to
$\phi=\|\bs\varphi\|$ 
with~$\phi \sim \chi(N)$. Moreover,~$a = \|\bs p-\bs q\|/\delta = \phi \|\tilde{\bs p}-\tilde{\bs q}\|/\delta = \alpha \phi$ with~$\alpha =
\|\tilde{\bs p}-\tilde{\bs q}\|/\delta$. Denoting the pdf of~$\chi(N)$ by~$\gamma_N(\phi) = c_N \phi^{N-1} \exp(-\tfrac{\phi^2}{2})$ and~$c_N =
2^{1-\frac{N}{2}} / \Gamma(\frac{N}{2})$, we can develop~$P_{s'}(\alpha, 1)=\
\int_0^{+\infty} P_{s'}(\alpha, 1|\phi)\,\gamma_N(\phi)\ \ud \phi$ as follows
\begin{align*}
P_{s'}(\alpha, 1)&\ts \leq\
\int_0^{\phi_0/\alpha} (1 - \tfrac{\alpha \phi}{2 \phi_0} (1-\lambda))\, \gamma_N(\phi)\,\ud
\phi\\
&\ts\quad +\  \int_{\phi_0/\alpha}^{+\infty} (\lambda + \tfrac{\phi_0}{2 \alpha \phi}(1-\lambda))\,\gamma_N(\phi)\,\ud \phi\\
&\ts = \lambda\ +\
(1-\lambda)\int_0^{\phi_0/\alpha} (1 - \tfrac{\alpha \phi}{2 \phi_0})\, \gamma_N(\phi)\,\ud
\phi\\
&\ts\qquad +\  (1-\lambda)\,\int_{\phi_0/\alpha}^{+\infty} \tfrac{\phi_0}{2
  \alpha \phi}\gamma_N(\phi)\,\ud \phi\\
&\ts = \lambda\ +\
(1-\lambda)\int_0^{+\infty} \varphi(\tfrac{\alpha \phi}{\phi_0})\,\gamma_N(\phi) \ud \phi,
\end{align*}
with~$\varphi(t) = 1-\tinv{2}t$ if~$0\leq t < 1$ and~$\varphi(t) =
\tinv{2t}$ if~$t\geq 1$. 

We can notice that~$t\varphi(t)$, which is equal to~$\tinv{2}t(2-t)$
over~$[0,1]$ and to~$\tinv{2}$ for~$t\geq 1$, is a concave function.
Therefore, by Jensen inequality, 
\begin{align*}
&\ts\int_0^{+\infty} \varphi(\tfrac{\alpha \phi}{\phi_0})\,\gamma_N(\phi)
\ud \phi\\
&\ts=\ \tfrac{c_N}{c_{N-1}}\,\int_0^{+\infty} \phi\,\varphi(\tfrac{\alpha
  \phi}{\phi_0})\,\gamma_{N-1}(\phi) \ud \phi\\
&\ts\leq\ \tfrac{c_N}{c_{N-1}} (\bb E_{\gamma_{N-1}} \phi)\, \varphi(\tfrac{\alpha
  \bb E_{\gamma_{N-1}}\phi}{\phi_0}).
\end{align*}
We have also~$c_{N}/c_{N-1} =
 \Gamma(\frac{N-1}{2})/(\sqrt 2\Gamma(\frac{N}{2}))$ and~$\bb
 E_{\gamma_{N-1}} \phi = \sqrt 2
\Gamma(\tfrac{N}{2})/\Gamma(\tfrac{N-1}{2}) = \sqrt{2\pi} \kappa_N$, so that
$\tfrac{c_N}{c_{N-1}} (\bb E_{\gamma_{N-1}} \phi) = 1$ and 
$$
\tfrac{\alpha
  \bb E_{\gamma_{N-1}}\phi}{\phi_0} = \alpha\,\tfrac {1-\lambda}{2\kappa_N}\, 
  \bb E_{\gamma_{N-1}}\phi = \sqrt{\tfrac {\pi}{2}} (1-\lambda) \alpha.
$$
Consequently, since~$\varphi(t) \leq \tfrac{2}{2+t}$,
\begin{align*}
P_{s'}(\alpha, 1)&\ts\leq \lambda + (1-\lambda) \varphi\big(\sqrt{\tfrac
  {\pi}{2}} (1-\lambda) \alpha\big)\\
&\ts\leq \lambda + (1-\lambda) \tfrac{2}{2 + \sqrt{\frac
  {\pi}{2}} (1-\lambda) \alpha}\\
&\ts = 1 - \tfrac{\sqrt{\frac
  {\pi}{2}} (1-\lambda)^2 \alpha}{2 + \sqrt{\frac
  {\pi}{2}} (1-\lambda) \alpha} =\ 1 - \tfrac{\sqrt{\frac{2}{\pi}} \alpha}{2 + \alpha}\\
&\ts <\ 1 - \tfrac{3 \alpha}{8 + 4\alpha},
\end{align*}
taking~$(1-\lambda) = \sqrt{\tfrac{2}{\pi}}> 3/4$.

Moreover, from the bounds on~$\kappa_N$ given in~\eqref{eq:bound-tau_N}, this shows
also that~$$
\tfrac{2s'}{\|\tilde{\bs p}-\tilde{\bs
  q}\|} = \tfrac{2r}{L} = \tfrac{\lambda}{2\kappa_N} \geq (1-\sqrt{\tfrac{2}{\pi}})
\sqrt{\tfrac{\pi}{2}}\tinv{(N-1)^{1/2}} > \tfrac{1}{4\sqrt N},
$$
as stated at the beginning of Lemma~\ref{lemma:main-proba-bound}.

\end{document}